\documentclass[aps,superscriptaddress,showkeys]{revtex4}
\usepackage{braket}
\usepackage{amsmath}
\usepackage{amssymb}
\usepackage{amsthm}
\usepackage{graphicx}
\usepackage[colorlinks,
            linkcolor=blue,
            anchorcolor=blue,
            citecolor=blue]{hyperref}
\usepackage{appendix}

\makeatletter 
\newcommand{\rmnum}[1]{\romannumeral #1}
\newcommand{\Rmnum}[1]{\expandafter\@slowromancap\romannumeral #1@}
\makeatother
\theoremstyle{definition}
\newtheorem{theorem}{Theorem}
\newtheorem{definition}[theorem]{Definition}
\newtheorem{lemma}[theorem]{Lemma}
\newtheorem{proposition}[theorem]{Proposition}
\newtheorem{corollary}[theorem]{Corollary}
\newtheorem{example}[theorem]{Example}

\begin{document}

\title{Coherence and imaginarity of qubit states under vanishing Bargmann invariants}

\author{Hao-Fan Wang}
\email{haofanwang@buaa.edu.cn}
\affiliation{LMIB(Beihang University), Ministry of Education, and School of Mathematical Sciences, Beihang University, Beijing 100191, China}
\author{Lin Chen}
\email{linchen@buaa.edu.cn(corresponding author)}
\affiliation{LMIB(Beihang University), Ministry of Education, and School of Mathematical Sciences, Beihang University, Beijing 100191, China}

\begin{abstract}
Set coherence and set imaginarity are two important quantum resources defined by sets of states. We investigate tuples of qubit states under vanishing higher-order Bargmann invariants and show that their set coherence and set imaginarity are each related to whether they contain a pair of orthogonal states. Specifically, when the relevant higher-order Bargmann invariants vanish, the tuple contains a pair of orthogonal states if and only if removing some state makes the remaining tuple incoherent. Moreover, when maximally mixed states are excluded, the same vanishing conditions imply that the tuple is imaginarity-free if and only if it contains a pair of orthogonal states. As a consequence, under these conditions, a tuple of single-qubit states has imaginarity precisely when every tuple obtained by deleting one of its states is coherent.
\end{abstract}

\keywords{Bargmann invariants, set coherence, set imaginarity, qubit states, orthogonal states}

\maketitle

\section{Introduction}
Bargmann invariants \cite{Bargmann0,Bargmann1}, also called multivariate traces of density operators, are recognized as fundamental local unitary invariants in quantum information theory, and they can be estimated directly with quantum circuits \cite{Oszmaniec,Quek,Simonov}. These multivariate traces of quantum states have various applications and links to many concepts, including multiphoton indistinguishability \cite{Menssen,Jones1,Jones2}, quantum error mitigation \cite{JMLiang}, weak values \cite{weakvalue}, Kirkwood-Dirac quasiprobability \cite{KDqp1,KDqp2}, quantum communication \cite{Elliott}, entanglement detection \cite{LinZhang}, set coherence \cite{setcoherence,Wagner1,YLWang}, and set imaginarity \cite{Fernandes,MSLi,YGuo}. This naturally motivates the investigation of the sets of values that Bargmann invariants can take. In a series of works, these sets have been completely characterized \cite{Bargset1,Bargset2,Bargset3,Bargset4}. Recently, the concepts of Bargmann scenarios and Bargmann polytopes have been introduced in Ref.\cite{Wagner1}, thereby providing a unified framework in which these sets appear as special cases.

Coherence and imaginarity are important features of quantum states, and they have been recognized as two crucial physical resources \cite{CoResource,ImResource,ResoTheory} in quantum information. Furthermore, precise operational interpretations have been shown for some single-state coherence measures \cite{omco} and for the single-state resource theory of imaginarity \cite{omim1,omim2}. However, the single-state resource theory overlooks many phenomena that call for a concrete understanding of the collective properties of quantum states, such as the collective behavior of a set of single photons entering a linear-optical multiport interferometer \cite{Shchesnovich1,Shchesnovich2,Jones1,Jones2} and the non-Markovianity of quantum channels \cite{nMarkov}. Motivated by this, quantum resources defined by sets of states have been increasingly studied. To extend single-state coherence, the notion of set coherence was introduced in Ref.\cite{setcoherence} and experimentally verified in Ref.\cite{setcoherence2}. In parallel, the notion of set imaginarity was introduced in Ref.\cite{setimaginarity}. Therefore, a natural question is how to determine set coherence and set imaginarity. Indeed, Ref.\cite{MSLi} has shown that the set imaginarity of qubit states can be fully characterized by second-order Bargmann invariants, and also by the imaginary parts of third-order Bargmann invariants. In addition, Ref.\cite{YLWang} has shown that the set coherence of states can be fully characterized by a fourth-order Bargmann scenario. Although these studies indicate that Bargmann invariants provide useful tools for determining set coherence and set imaginarity, the connection between set coherence and set imaginarity remains unclear.\par

In this work, we investigate how vanishing higher-order Bargmann invariants constrain set coherence and set imaginarity of an $n$-tuple of single-qubit states, where $n\geq 4$. In fact, we find that orthogonality plays a crucial role. Our main results are as follows. First, we establish three theorems connecting set coherence to orthogonality. Proposition \ref{pp:rho1rho2rho3rho4} deals with the case $n=4$: if the three fourth-order Bargmann invariants $\operatorname{Tr}(\rho_{1}\rho_{2}\rho_{3}\rho_{4})$, $\operatorname{Tr}(\rho_{1}\rho_{3}\rho_{4}\rho_{2})$ and $\operatorname{Tr}(\rho_{1}\rho_{4}\rho_{2}\rho_{3})$ vanish, then the tuple $(\rho_{1},\rho_{2},\rho_{3},\rho_{4})$ contains a pair of orthogonal states if and only if there exists a state whose removal makes the remaining tuple incoherent. Proposition \ref{thm_5rho} extends this equivalence to $n=5$, where the vanishing conditions are formulated through the Bargmann scenario associated with the alternating group $A_5$. Theorem \ref{thm_nrho} further extends it to all $n \geq 6$, with the vanishing conditions given by the Bargmann scenario associated with the symmetric group $S_n$. Second, we establish two theorems connecting set imaginarity to orthogonality. Proposition \ref{4imaginarity} treats the case $n=4$: if the tuple $(\rho_{1},\rho_{2},\rho_{3},\rho_{4})$ contains no maximally mixed state and satisfies the same vanishing conditions as in Proposition \ref{pp:rho1rho2rho3rho4}, then the tuple $(\rho_{1},\rho_{2},\rho_{3},\rho_{4})$ is imaginarity-free if and only if it contains a pair of orthogonal states. Theorem \ref{nimaginarity} generalizes this equivalence to all $n\geq 4$, with the vanishing conditions given by the Bargmann scenario associated with the symmetric group $S_n$. These results indicate that, under the conditions above, an $n$-tuple of single-qubit states has imaginarity if and only if every tuple obtained by removing one of its states is coherent.\par

The remainder of this paper is structured as follows: In Sec. \ref{sec:pre}, we clarify the fundamental concepts and main knowledge used in this work. In Sec. \ref{sec3}, we characterize the connection between set coherence and orthogonality under vanishing Bargmann invariants. In Sec. \ref{sec4}, we characterize the connection between set imaginarity and orthogonality under vanishing Bargmann invariants, and establish the connection between set coherence and set imaginarity via Corollary \ref{maincor}. Finally, in Sec. \ref{sec5}, we summarize our results and provide discussion. \par

\section{Preliminaries}
\label{sec:pre}
In this section, we introduce the main facts and tools used in this paper. We recall the definition and some properties of Bargmann scenarios. We regard $\{1,\cdots, n\}$ as an alphabet whose elements are called letters, and let finite sequences of letters be called words. Let $w=(l_{1},\cdots,l_{m})$ be a word over the alphabet $\{1,\cdots,n\}$. Here the letters $l_{1},\cdots,l_{m}$ need not be different; that is, repetitions are allowed. Let $\vec{\rho}=(\rho_{1},\cdots,\rho_{n})$ be an $n$-tuple of density operators on $\mathbb{C}^{d}$. The generalized Bargmann invariant associated with $\vec{\rho}$ and $w$ is defined by 
\begin{equation}\label{eq:Delta w}
    \Delta_{w}(\vec{\rho}):=\operatorname{Tr}(\rho_{l_{1}}\cdots\rho_{l_{m}}).
\end{equation}\par
We say that two words are cyclically equivalent if one can be obtained from the other by applying a finite number of cyclic permutations. Clearly, for a fixed $n$-tuple $\vec{\rho}$, cyclically equivalent words give the same Bargmann invariant. \par
    
\begin{definition}\label{de:Bargmann scenario}
    \cite{Wagner1,YLWang} A {\bf Bargmann scenario} is a finite set $\mathcal{W}$ of words with its defining value map
    \begin{equation}
        \label{eq:Bargmann scenario}\Phi_{\mathcal{W}}(\vec{\rho}):= (\Delta_{w}(\vec{\rho}))_{w\in\mathcal{W}}.
    \end{equation}
\end{definition}

For any two density operators on qubit system, a fundamental property can be stated as follows.

\begin{lemma}\label{pp:rho sigma}
    Let $\rho,\sigma$ be two density operators on $\mathbb{C}^{2}$. Suppose $\mathcal{W}$ is the Bargmann scenario consisting of all words over the alphabet $\{1,2\}$. Then $\Delta_{w}(\rho,\sigma)\geq 0$ for all $w\in\mathcal{W}$. Moreover, $\rho\sigma=0$ if and only if there exists some $w\in\mathcal{W}$ such that $\Delta_{w}(\rho,\sigma)=0$.    
\end{lemma}
\begin{proof}
    See the proof of Lemma \ref{pp:rho sigma} in Appendix \ref{Appendix pp:rho sigma}.
\end{proof}

The first conclusion of Lemma \ref{pp:rho sigma} cannot be generalized to higher dimensions. For example, take $\rho=\begin{pmatrix} \frac{1}{2} & 0 & 0\\ 0 & \frac{1}{3} & 0\\ 0 & 0 & \frac{1}{6} \end{pmatrix}, \sigma=\begin{pmatrix} \frac{1}{3} & \frac{1}{6} & \frac{1}{6}\\ \frac{1}{6} & \frac{1}{3} & \frac{\mathrm{i}}{6}\\ \frac{1}{6} & -\frac{\mathrm{i}}{6} & \frac{1}{3} \end{pmatrix}$. Then $\Delta_{112212}(\rho,\sigma)=\operatorname{Tr}(\rho^{2}\sigma^{2}\rho\sigma)=\frac{13}{972}-\frac{\mathrm{i}}{23328}\notin\mathbb{R}$. In fact, Johnson and Hillar \cite{Johnson} conjectured that, for any word $w$ over the alphabet $\{1,2\}$, the inequality $\Delta_{w}(A,B)\geq 0$ holds for every pair of positive semidefinite operators $A,B$ if and only if $w$ is nearly symmetric. 

Next, we recall the definition and some properties of set imaginarity.

\begin{definition}\label{de:imaginary-free}
    \cite{setimaginarity,MSLi}
    Let $\vec{\rho}=(\rho_{1},\cdots,\rho_{n})$ be an $n$-tuple of density operators where $\rho_{i}\in \mathcal{D}(\mathbb{C}^{d})$. We say that $\vec{\rho}$ is {\bf imaginarity-free} if there exists a unitary operator $U$ acting on $\mathbb{C}^{d}$ such that $U\rho_{k}U^{\dagger}$ is real with respect to some basis $\{\ket{j}\}_{j=1}^{d}$ for all $k\in\{1,2,\cdots,n\}$. Otherwise, we say that $\vec{\rho}$ has imaginarity.
\end{definition}

Given an $n$-tuple of single-qubit states $\vec{\rho}=(\rho_{1},\cdots,\rho_{n})\in\mathcal{D}(\mathbb{C}^{2})^{n}$, the Gram matrix associated to the related tuple of Bloch vectors $\boldsymbol{r}_{\vec{\rho}}=(\boldsymbol{r}_{1},\cdots,\boldsymbol{r}_{n})$ is defined by $G_{\boldsymbol{r}_{\vec{\rho}}}=(\braket{\boldsymbol{r}_{k},\boldsymbol{r}_{l}})_{n\times n}$, where $\boldsymbol{r}_{k}$ is the Bloch vector of the state $\rho_{k}$. If we denote $B=(\boldsymbol{r}_{1},\cdots,\boldsymbol{r}_{n})$, then we can see that $\operatorname{rank}(G_{\boldsymbol{r}_{\vec{\rho}}})=\operatorname{rank}(B^{\mathsf{T}}B)=\operatorname{rank}(B)\leq 3$. The necessary and sufficient condition for single-qubit multistates $\vec{\rho}$ to have imaginarity can be stated as follows.

\begin{lemma}\label{th:imaginarity}
    (Theorem III.3. in Ref.\cite{MSLi}) Let $\vec{\rho}=(\rho_{1},\cdots,\rho_{n})$ be an $n$-tuple of single-qubit states where $\rho_{i}\in \mathcal{D}(\mathbb{C}^{2})$. Then $\vec{\rho}$ has imaginarity if and only if ${\rm rank}(G_{\boldsymbol{r}_{\vec{\rho}}})=3$.
\end{lemma}

Finally, we recall the definition of set coherence. 

\begin{definition}
    \cite{setcoherence,MSLi}
    Let $\vec{\rho}=(\rho_{1},\cdots,\rho_{n})$ be an $n$-tuple of density operators where $\rho_{i}\in \mathcal{D}(\mathbb{C}^{d})$. We say that $\vec{\rho}$ is {\bf incoherent} if there exists a unitary operator $U$ acting on $\mathbb{C}^{d}$ such that $U\rho_{k}U^{\dagger}$ is diagonal with respect to some basis $\{\ket{j}\}_{j=1}^{d}$ for all $k\in\{1,2,\cdots,n\}$. Otherwise, we say that $\vec{\rho}$ is {\bf coherent}.
\end{definition}

As is well known, an $n$-tuple of states $(\rho_{1},\cdots,\rho_{n})$ is incoherent if and only if $[\rho_{i},\rho_{j}]=0$ for all $i,j\in\{1,2,\cdots,n\}$.

\section{The connection between coherence and orthogonality}\label{sec3}

In this section, we always assume that $\vec{\rho}=(\rho_{1},\cdots,\rho_{n})$ ($n\geq 4$) is a tuple of single-qubit states and its Bargmann invariants $\Delta_{w}(\vec{\rho})$ vanish for certain words $w$ of length $n$. Under this assumption, we show that the tuple $\vec{\rho}=(\rho_{1},\cdots,\rho_{n})$ contains a pair of orthogonal states if and only if there exists some $s\in\{1,2,\cdots,n\}$ such that the tuple $(\rho_{k})_{k\neq s}$ is incoherent. We present Proposition \ref{pp:rho1rho2rho3rho4} and Proposition \ref{thm_5rho} for the cases $n=4$ and $n=5$, respectively, and propose Theorem \ref{thm_nrho} for $n\geq 6$. 

\begin{proposition}\label{pp:rho1rho2rho3rho4}
    Let $\rho_{1},\rho_{2},\rho_{3},\rho_{4}$ be four density operators on $\mathbb{C}^{2}$. Suppose that 
    \begin{equation}\label{key condition of pp:4qubit}
        \operatorname{Tr}(\rho_{1}\rho_{2}\rho_{3}\rho_{4})= \operatorname{Tr}(\rho_{1}\rho_{3}\rho_{4}\rho_{2})= \operatorname{Tr}(\rho_{1}\rho_{4}\rho_{2}\rho_{3})=0.
    \end{equation}
    Then the following statements are equivalent:\\
    (\rmnum{1}) There exist some $i,j\in\{1,2,3,4\}$ such that $\rho_{i}\rho_{j}=0$.\\
    (\rmnum{2}) There exists some $s\in\{1,2,3,4\}$ such that the tuple $(\rho_{k})_{k\neq s}$ is incoherent (that is, $[\rho_{k},\rho_{l}]=0$ for all $k,l\in\{1,2,3,4\}\backslash\{s\}$).
\end{proposition}
\begin{proof}
    (\rmnum{1})$\Rightarrow$(\rmnum{2}). Assume that $\{1,2,3,4\}\backslash\{i,j\}=\{k,l\}$. Due to $\rho_{i}\rho_{j}=0$, we can choose a basis $\ket{0},\ket{1}$ in which 
    \begin{equation}
        \rho_{i}=\begin{pmatrix} 1 & 0 \\ 0 & 0 \end{pmatrix},
        \rho_{j}=\begin{pmatrix} 0 & 0 \\ 0 & 1 \end{pmatrix},
        \rho_{k}=\begin{pmatrix} p & \beta\\ \overline{\beta} & 1-p \end{pmatrix},
        \rho_{l}=\begin{pmatrix} q & \gamma\\ \overline{\gamma} & 1-q \end{pmatrix},
    \end{equation}
    where $p,q\in [0,1]$, $\beta,\gamma\in\mathbb{C}$, $|\beta|^{2}\leq p(1-p)$ and $|\gamma|^{2}\leq q(1-q)$. From \eqref{key condition of pp:4qubit} and cyclicity of the trace, we have 
    \begin{flalign*}
        \operatorname{Tr}(\rho_{2}\rho_{3}\rho_{4}\rho_{1})=0, \operatorname{Tr}(\rho_{2}\rho_{1}\rho_{3}\rho_{4})=0, \operatorname{Tr}(\rho_{2}\rho_{3}\rho_{1}\rho_{4})=0,\\
        \operatorname{Tr}(\rho_{3}\rho_{4}\rho_{1}\rho_{2})=0, \operatorname{Tr}(\rho_{3}\rho_{4}\rho_{2}\rho_{1})=0, \operatorname{Tr}(\rho_{3}\rho_{1}\rho_{4}\rho_{2})=0,\\
        \operatorname{Tr}(\rho_{4}\rho_{1}\rho_{2}\rho_{3})=0, \operatorname{Tr}(\rho_{4}\rho_{2}\rho_{1}\rho_{3})=0, \operatorname{Tr}(\rho_{4}\rho_{2}\rho_{3}\rho_{1})=0.
    \end{flalign*}
    Therefore, we obtain $0=\operatorname{Tr}(\rho_{i}\rho_{k}\rho_{j}\rho_{l})=\beta\bar{\gamma}$ or $0=\operatorname{Tr}(\rho_{i}\rho_{l}\rho_{j}\rho_{k})=\gamma\bar{\beta}$. Hence $\beta=0$ or $\gamma=0$. If $\beta=0$, we have $[\rho_{l_{1}},\rho_{l_{2}}]=0$ for all $l_{1},l_{2}\in\{i,j,k\}$. If $\gamma=0$, we have $[\rho_{l_{1}},\rho_{l_{2}}]=0$ for all $l_{1},l_{2}\in\{i,j,l\}$.\par 
    \noindent
    (\rmnum{2})$\Rightarrow$(\rmnum{1}). Assume that $\{1,2,3,4\}\backslash\{s\}=\{k_{1},k_{2},k_{3}\}$, where $k_{1}<k_{2}<k_{3}$. Due to $[\rho_{k},\rho_{l}]=0$ for all $k,l\in\{k_{1},k_{2},k_{3}\}$, we can choose a basis $\ket{0},\ket{1}$ in which 
    \begin{equation}
        \rho_{k_{1}}=\begin{pmatrix} a_{1} & 0 \\ 0 & 1-a_{1} \end{pmatrix},
        \rho_{k_{2}}=\begin{pmatrix} a_{2} & 0\\ 0 & 1-a_{2} \end{pmatrix},
        \rho_{k_{3}}=\begin{pmatrix} a_{3} & 0\\ 0 & 1-a_{3} \end{pmatrix},
        \rho_{s}=\begin{pmatrix} p & \beta\\ \overline{\beta} & 1-p \end{pmatrix},
    \end{equation}
    where $a_{1},a_{2},a_{3},p\in [0,1]$, $\beta\in \mathbb{C}$ and $|\beta|^{2}\leq p(1-p)$. From \eqref{key condition of pp:4qubit} and cyclicity of the trace, we have 
    $\operatorname{Tr}(\rho_{2}\rho_{3}\rho_{4}\rho_{1})=0$, $\operatorname{Tr}(\rho_{1}\rho_{3}\rho_{4}\rho_{2})=0$, $\operatorname{Tr}(\rho_{1}\rho_{4}\rho_{2}\rho_{3})=0$. With help of this and commutativity, we obtain $0=\operatorname{Tr}(\rho_{k_{1}}\rho_{k_{2}}\rho_{k_{3}}\rho_{s})=
    p\prod\limits_{k=1}^{3}a_{k}+(1-p)\prod\limits_{k=1}^{3}(1-a_{k})$. Therefore, we have $p\prod\limits_{k=1}^{3}a_{k}=0$ and $(1-p)\prod\limits_{k=1}^{3}(1-a_{k})=0$. If $p\in (0,1)$, then 
$\prod\limits_{k=1}^{3}a_{k}=0$ and $\prod\limits_{k=1}^{3}(1-a_{k})=0$. So there exist $i,j\in\{1,2,3\}(i\neq j)$ such that $a_{i}=0$ and $a_{j}=1$. This implies $\rho_{k_{i}}\rho_{k_{j}}=0$. On the other hand, if $p=0$ or $p=1$, then $\beta=0$ from $|\beta|^{2}\leq p(1-p)$. So there exists some $i\in\{1,2,3\}$ such that $a_{i}=1$ or $a_{i}=0$. Hence $\rho_{k_{i}}\rho_{s}=0$.
\end{proof}

Proposition \ref{pp:rho1rho2rho3rho4} can be used to verify set coherence, and the following is an example.

\begin{example}\label{example1}
    Consider the following four density operators on $\mathbb{C}^{2}$:
    \begin{equation*}
    \begin{split}
    \rho_{1}=\frac{1}{2}\begin{pmatrix} 1 & \frac{1}{\sqrt{3}}\\ \frac{1}{\sqrt{3}} & 1 \end{pmatrix} &, 
    \rho_{2}=\frac{1}{2}\begin{pmatrix} 1 & -\frac{i}{\sqrt{3}}\\ \frac{i}{\sqrt{3}} & 1 \end{pmatrix}, \\
    \rho_{3}=\frac{1}{2}\begin{pmatrix} 1+\frac{1}{\sqrt{3}} & 0\\ 0 & 1-\frac{1}{\sqrt{3}}\end{pmatrix} &,
    \rho_{4}=\frac{1}{2}\begin{pmatrix} 1-\frac{1}{\sqrt{3}} & -\frac{1-i}{\sqrt{3}}\\ -\frac{1+i}{\sqrt{3}} & 1+\frac{1}{\sqrt{3}}\end{pmatrix}.
    \end{split}
    \end{equation*}
    By direct computation, we have $\operatorname{Tr}(\rho_{1}\rho_{2}\rho_{3}\rho_{4})=0$, $\operatorname{Tr}(\rho_{1}\rho_{3}\rho_{4}\rho_{2})=0$, $\operatorname{Tr}(\rho_{1}\rho_{4}\rho_{2}\rho_{3})=0$, and $\operatorname{Tr}(\rho_{i}\rho_{j})>0$ for all $i,j\in\{1,2,3,4\}$. Combining this with Proposition \ref{pp:rho1rho2rho3rho4}, we obtain that all of the following tuples are coherent:
    \begin{equation*}
        (\rho_{1},\rho_{2},\rho_{3},\rho_{4}), (\rho_{2},\rho_{3},\rho_{4}), (\rho_{1},\rho_{3},\rho_{4}), (\rho_{1},\rho_{2},\rho_{4}), (\rho_{1},\rho_{2},\rho_{3}). 
    \end{equation*}
    In fact, direct computation shows that $[\rho_{i},\rho_{j}]\neq 0$ for all distinct $i,j\in\{1,2,3,4\}$. 
\end{example}

   For $n=5$, we introduce the following proposition based on the alternating group $A_{5}$. 

\begin{proposition}\label{thm_5rho}
    Let $\rho_{1},\rho_{2},\rho_{3},\rho_{4},\rho_{5}$ be five density operators on $\mathbb{C}^{2}$. Denote the Bargmann scenario $\mathcal{W}=\{(\sigma(1),\sigma(2),\sigma(3),\sigma(4),\sigma(5))|\sigma\in A_{5}\enspace\mbox{and}\enspace\sigma(1)=1\}$, where $A_{5}$ is the alternating group on $\{1,2,3,4,5\}$. Suppose that 
    \begin{equation}\label{key condition of pp:5qubit}
        \Delta_{w}(\rho_{1},\rho_{2},\rho_{3},\rho_{4},\rho_{5})=0,\enspace \forall w\in\mathcal{W}.
    \end{equation}
    Then the following statements are equivalent:\\
    (\rmnum{1}) There exist some $i,j\in\{1,2,3,4,5\}$ such that $\rho_{i}\rho_{j}=0$.\\
    (\rmnum{2}) There exists some $s\in\{1,2,3,4,5\}$ such that the tuple $(\rho_{k})_{k\neq s}$ is incoherent (that is, $[\rho_{k},\rho_{l}]=0$ for all $k,l\in\{1,2,3,4,5\}\backslash\{s\}$).
\end{proposition}
\begin{proof}
    (\rmnum{1})$\Rightarrow$(\rmnum{2}). Assume that $\{1,2,3,4,5\}\backslash\{i,j\}=\{k_{1},k_{2},k_{3}\}$. Due to $\rho_{i}\rho_{j}=0$, we can choose a basis $\ket{0},\ket{1}$ in which 
    \begin{equation}\label{eq:rho_i}
        \rho_{i}=\begin{pmatrix} 1 & 0 \\ 0 & 0 \end{pmatrix},
        \rho_{j}=\begin{pmatrix} 0 & 0 \\ 0 & 1 \end{pmatrix},
        \rho_{k_{l}}=\begin{pmatrix} p_{l} & \alpha_{l}\\ \overline{\alpha_{l}} & 1-p_{l} \end{pmatrix}, l=1,2,3
    \end{equation}
where $p_{l}\in [0,1]$, $\alpha_{l}\in \mathbb{C}$ and $|\alpha_{l}|^{2}\leq p_{l}(1-p_{l})$. Due to cyclicity of the trace, we see that \eqref{key condition of pp:5qubit} is equivalent to $\operatorname{Tr}(\rho_{\sigma(1)}\rho_{\sigma(2)}\rho_{\sigma(3)}\rho_{\sigma(4)}\rho_{\sigma(5)})=0$ for all $\sigma\in A_{5}$. Denote $\sigma_{1}=(2\enspace3)(4\enspace5),\sigma_{2}=(2\enspace 5\enspace 3),\sigma_{3}=(2\enspace 4\enspace 3)$. Hence $\sigma_{1},\sigma_{2},\sigma_{3}\in A_{5}$, and we can see that
    \begin{equation}
        \begin{split}
            (\sigma_{1}(1),\sigma_{1}(2),\sigma_{1}(3),\sigma_{1}(4),\sigma_{1}(5))=(1,3,2,5,4),\\
            (\sigma_{2}(1),\sigma_{2}(2),\sigma_{2}(3),\sigma_{2}(4),\sigma_{2}(5))=(1,5,2,4,3),\\
            (\sigma_{3}(1),\sigma_{3}(2),\sigma_{3}(3),\sigma_{3}(4),\sigma_{3}(5))=(1,4,2,3,5).
        \end{split}
    \end{equation}
Because there exists some $\tau\in A_{5}$ such that $\tau(1)=i$ and $\tau(2)=j$, we obtain
\begin{equation}
        \begin{split}
            (\tau_{1}(1),\tau_{1}(2),\tau_{1}(3),\tau_{1}(4),\tau_{1}(5))=(i,\tau(3),j,\tau(5),\tau(4)),\\
            (\tau_{2}(1),\tau_{2}(2),\tau_{2}(3),\tau_{2}(4),\tau_{2}(5))=(i,\tau(5),j,\tau(4),\tau(3)),\\
            (\tau_{3}(1),\tau_{3}(2),\tau_{3}(3),\tau_{3}(4),\tau_{3}(5))=(i,\tau(4),j,\tau(3),\tau(5)),
        \end{split}
    \end{equation}
    where $\tau_{k}:=\tau\circ\sigma_{k}\in A_{5}$ for all $k\in\{1,2,3\}$. Due to $\{\tau(3),\tau(5),\tau(4)\}=\{1,2,3,4,5\}\backslash\{i,j\}$, we can take $k_{1}=\tau(3),k_{2}=\tau(5),k_{3}=\tau(4)$. Hence we have
\begin{equation}
    \label{eq:0=Tr rhoi}
        \begin{cases}
        0=\operatorname{Tr}(\rho_{i}\rho_{k_{1}}\rho_{j}\rho_{k_{2}}\rho_{k_{3}})=\alpha_{1}(\overline{\alpha_{2}}p_{3}+(1-p_{2})\overline{\alpha_{3}}),\\
        0=\operatorname{Tr}(\rho_{i}\rho_{k_{2}}\rho_{j}\rho_{k_{3}}\rho_{k_{1}})=\alpha_{2}(\overline{\alpha_{3}}p_{1}+(1-p_{3})\overline{\alpha_{1}}),\\
        0=\operatorname{Tr}(\rho_{i}\rho_{k_{3}}\rho_{j}\rho_{k_{1}}\rho_{k_{2}})=\alpha_{3}(\overline{\alpha_{1}}p_{2}+(1-p_{1})\overline{\alpha_{2}}).
        \end{cases}
    \end{equation}
We claim that at most one of $\alpha_{1},\alpha_{2},\alpha_{3}$ is non-zero. Two cases are distinguished. If $\alpha_{l}\neq 0$ for all $l\in\{1,2,3\}$, then 
    \begin{equation*}
        \begin{cases}
        \overline{\alpha_{2}}p_{3}+(1-p_{2})\overline{\alpha_{3}}=0,\\
        \overline{\alpha_{3}}p_{1}+(1-p_{3})\overline{\alpha_{1}}=0,\\
        \overline{\alpha_{1}}p_{2}+(1-p_{1})\overline{\alpha_{2}}=0.
        \end{cases}
    \end{equation*}
So we have $0=\det\begin{pmatrix} 0 & p_{3} & 1-p_{2}\\ 1-p_{3} & 0 & p_{1}\\ p_{2} & 1-p_{1} & 0 \end{pmatrix}=p_{1}p_{2}p_{3}+(1-p_{1})(1-p_{2})(1-p_{3})$. However, due to $|\alpha_{l}|^{2}\leq p_{l}(1-p_{l})$ and $\alpha_{l}\neq 0$, we have $p_{l}\in (0,1)$ for all $l\in\{1,2,3\}$, a contradiction. 

We consider another case, namely exactly two elements of $(\alpha_{1},\alpha_{2},\alpha_{3})$ are non-zero. Using \eqref{eq:rho_i} and \eqref{eq:0=Tr rhoi}, we assume without loss of generality that $\alpha_{1}\neq 0$, $\alpha_{2}\neq 0$, and $\alpha_{3}=0$. In this case, we have 
    \begin{flalign}
        \alpha_{1}\overline{\alpha_{2}}p_{3}=\alpha_{1}(\overline{\alpha_{2}}p_{3}+(1-p_{2})\overline{\alpha_{3}})=0,\label{eq2.1 5qubit}\\
        (1-p_{3})\overline{\alpha}_{1}\alpha_{2}=\alpha_{2}(\overline{\alpha_{3}}p_{1}+(1-p_{3})\overline{\alpha_{1}})=0\label{eq2.2 5qubit}.
    \end{flalign}
    Therefore, we obtain $p_{3}=0$ from (\ref{eq2.1 5qubit}) and $p_{3}=1$ from (\ref{eq2.2 5qubit}), a contradiction. We have proven the claim below \eqref{eq:0=Tr rhoi}.\par
    
Consequently, there exist some $s,t\in\{1,2,3\}$ such that $\alpha_{s}=\alpha_{t}=0$. As a result using \eqref{eq:rho_i}, we have $[\rho_{k},\rho_{l}]=0$ for all $k,l\in\{i,j,k_{s},k_{t}\}$. \par

(\rmnum{2})$\Rightarrow$(\rmnum{1}). Assume that $\{1,2,3,4,5\}\backslash\{s\}=\{k_{1},k_{2},k_{3},k_{4}\}$, where $k_{1}<k_{2}<k_{3}<k_{4}$. Due to $[\rho_{k},\rho_{l}]=0$ for all $k,l\in\{k_{1},k_{2},k_{3},k_{4}\}$, we can choose a basis $\ket{0},\ket{1}$ in which
    \begin{equation}
        \rho_{s}=\begin{pmatrix} p & \beta\\ \overline{\beta} & 1-p \end{pmatrix},
        \rho_{k_{l}}=\begin{pmatrix} a_{l} & 0 \\ 0 & 1-a_{l} \end{pmatrix}, l=1,2,3,4
    \end{equation}
    where $a_{l},p\in [0,1]$, $\beta\in \mathbb{C}$ and $|\beta|^{2}\leq p(1-p)$. Due to cyclicity of the trace, we see that \eqref{key condition of pp:5qubit} is equivalent to     
    $\operatorname{Tr}(\rho_{\sigma(1)}\rho_{\sigma(2)}\rho_{\sigma(3)}\rho_{\sigma(4)}\rho_{\sigma(5)})=0$ for all $\sigma\in A_{5}$. With help of this and commutativity, we obtain $0=\operatorname{Tr}(\rho_{k_{1}}\rho_{k_{2}}\rho_{k_{3}}\rho_{k_{4}}\rho_{s})=p\prod\limits_{l=1}^{4}a_{l}+(1-p)\prod\limits_{l=1}^{4}(1-a_{l})$. Therefore, we obtain $p\prod\limits_{l=1}^{4}a_{l}=0$ and $(1-p)\prod\limits_{l=1}^{4}(1-a_{l})=0$. If $p\in (0,1)$, then $\prod\limits_{l=1}^{4}a_{l}=0$ and $\prod\limits_{l=1}^{4}(1-a_{l})=0$. So there exist $i,j\in\{1,2,3,4\}(i\neq j)$ such that $a_{i}=0$ and $a_{j}=1$. This implies $\rho_{k_{i}}\rho_{k_{j}}=0$. If $p=0$ or $p=1$, then $\beta=0$ from $|\beta|^{2}\leq p(1-p)$. Clearly, there exists some $i\in\{1,2,3,4\}$ such that $a_{i}=1$ or $a_{i}=0$. Hence $\rho_{k_{i}}\rho_{s}=0$. 
\end{proof}

    Now, we propose the following theorem for the general case $n\geq 6$.

\begin{theorem}\label{thm_nrho}
    Let $\vec{\rho}=(\rho_{1},\cdots,\rho_{n})$ be an $n$-tuple of single-qubit states where $\rho_{i}\in \mathcal{D}(\mathbb{C}^{2})$ and $n\geq 6$. Denote the Bargmann scenario $\mathcal{W}=\{(\sigma(1),\sigma(2),\cdots,\sigma(n))|\sigma\in S_{n}\enspace\mbox{and}\enspace\sigma(1)=1\}$, where $S_{n}$ is the symmetric group on $\{1,2,\cdots,n\}$. Suppose that 
    \begin{equation}\label{key condition of pp:nqubit}
       \Delta_{w}(\vec{\rho})=0,\enspace\forall w\in\mathcal{W}.
    \end{equation}
    Then the following statements are equivalent:\\
    (\rmnum{1}) There exist some $i,j\in\{1,2,\cdots,n\}$ such that $\rho_{i}\rho_{j}=0$.\\
    (\rmnum{2}) There exists some $s\in\{1,2,\cdots,n\}$ such that the tuple $(\rho_{k})_{k\neq s}$ is incoherent (that is, $[\rho_{k},\rho_{l}]=0$ for all $k,l\in\{1,2,\cdots,n\}\backslash\{s\}$).
\end{theorem}
\begin{proof}
    We prove Theorem \ref{thm_nrho} in Appendix \ref{Appendix thm_nrho}. 
\end{proof}


\section{The connection between imaginarity and orthogonality}\label{sec4}

In this section, we retain the same assumptions as those in Sec.~\ref{sec3} and additionally assume that $\operatorname{Tr}(\rho_{k}^{2})>\frac{1}{2}$ for all $k\in\{1,2,\cdots,n\}$. Under this assumption, we show that the tuple $\vec{\rho}=(\rho_{1},\cdots,\rho_{n})$ contains a pair of orthogonal states if and only if the tuple $\vec{\rho}$ is imaginarity-free. We present Proposition \ref{4imaginarity} for $n=4$, and propose Theorem \ref{nimaginarity} for $n\geq 4$. 

\begin{proposition}\label{4imaginarity}
     Let $\rho_{1},\rho_{2},\rho_{3},\rho_{4}$ be four density operators on $\mathbb{C}^{2}$. Suppose that 
    \begin{equation}\label{keyeq-4imaginarity}
        \operatorname{Tr}(\rho_{1}\rho_{2}\rho_{3}\rho_{4})=0, \operatorname{Tr}(\rho_{1}\rho_{3}\rho_{4}\rho_{2})=0, \operatorname{Tr}(\rho_{1}\rho_{4}\rho_{2}\rho_{3})=0,
    \end{equation}
    and $\operatorname{Tr}(\rho_{k}^{2})>\frac{1}{2}$ for all $k\in\{1,2,3,4\}$. Then the following statements are equivalent:\\
    (\rmnum{1}) There exist some $i,j\in\{1,2,3,4\}$ such that $\operatorname{Tr}(\rho_{i}\rho_{j})=0$.\\
    (\rmnum{2}) The tuple $\vec{\rho}=(\rho_{1},\rho_{2},\rho_{3},\rho_{4})$ is imaginarity-free.
\end{proposition}
\begin{proof}
    We write the four states in Bloch form,     
    \begin{equation}
        \rho_{k}=\frac{\mathbb{I}_{2}+\boldsymbol{r}_{k}\cdot\boldsymbol{\sigma}}{2}
    \end{equation}
    with $|\boldsymbol{r}_{k}|\leq 1$ ($k=1,2,3,4$). For all $k\in\{1,2,3,4\}$, we have $\boldsymbol{r}_{k}\neq 0$ because $\operatorname{Tr}(\rho_{k}^{2})\neq\frac{1}{2}$.\par
    \noindent
    (\rmnum{1})$\Rightarrow$(\rmnum{2}). Note that $\operatorname{Tr}(\rho_{i}\rho_{j})=0$ is equivalent to $\rho_{i}\rho_{j}=0$. From Proposition \ref{pp:rho1rho2rho3rho4}, there exists some $s\in\{1,2,3,4\}$ such that $[\rho_{k},\rho_{l}]=0$ for all $k,l\in\{1,2,3,4\}\backslash\{s\}$.  So at least three vectors in the sequence $\{\boldsymbol{r}_{k}\}_{k=1}^{4}$ are collinear. Denote $B=(\boldsymbol{r}_{1},\boldsymbol{r}_{2},\boldsymbol{r}_{3},\boldsymbol{r}_{4})$. Then we have ${\rm rank}(G_{\boldsymbol{r}_{\vec{\rho}}})={\rm rank}(B^{\mathsf{T}}B)={\rm rank}(B)\leq 2$ where $G_{\boldsymbol{r}_{\vec{\rho}}}=(\braket{\boldsymbol{r}_{k},\boldsymbol{r}_{l}})_{4\times 4}=B^TB$. Therefore, the tuple $\vec{\rho}=(\rho_{1},\rho_{2},\rho_{3},\rho_{4})$ is imaginarity-free from Lemma \ref{th:imaginarity}.\par
    
    \noindent
    (\rmnum{2})$\Rightarrow$(\rmnum{1}). Since $\vec{\rho}=(\rho_{1},\rho_{2},\rho_{3},\rho_{4})$ is imaginarity-free, we can assume without loss of generality that $\rho_{k}$ is real for all $k\in\{1,2,3,4\}$. In this case, $\operatorname{Tr}(\rho_{k}\sigma_{y})=0$ for all $k\in\{1,2,3,4\}$. Consequently, there exist $r_{k,1},r_{k,2}\in\mathbb{R}$ such that $\boldsymbol{r}_{k}=(r_{k,1},0,r_{k,2})$. Denote $\boldsymbol{u}_{k}=(r_{k,1},r_{k,2})$ for all $k\in\{1,2,3,4\}$. Clearly, $\boldsymbol{u}_{k}\neq 0$ because $\boldsymbol{r}_{k}\neq 0$. As is well known, we have the following fact
    \begin{equation}
        (\boldsymbol{a}\cdot\boldsymbol{\sigma})(\boldsymbol{b}\cdot\boldsymbol{\sigma})=(\boldsymbol{a}\cdot\boldsymbol{b})\mathbb{I}_{2}+i(\boldsymbol{a}\times\boldsymbol{b})\cdot\boldsymbol{\sigma}. 
    \end{equation}
    So we have     
    \begin{flalign*}
        8\cdot\operatorname{Tr}(\rho_{1}\rho_{2}\rho_{3}\rho_{4}) &=
         1+\sum\limits_{i<j}\boldsymbol{r}_{i}\cdot\boldsymbol{r}_{j}+(\boldsymbol{r}_{1}\cdot\boldsymbol{r}_{2})(\boldsymbol{r}_{3}\cdot\boldsymbol{r}_{4})
         -(\boldsymbol{r}_{1}\cdot\boldsymbol{r}_{3})(\boldsymbol{r}_{2}\cdot\boldsymbol{r}_{4})+(\boldsymbol{r}_{1}\cdot\boldsymbol{r}_{4})(\boldsymbol{r}_{2}\cdot\boldsymbol{r}_{3})\nonumber\\
         &\quad +i[(\boldsymbol{r}_{3}+\boldsymbol{r}_{4})\cdot(\boldsymbol{r}_{1}\times\boldsymbol{r}_{2})+(\boldsymbol{r}_{1}+\boldsymbol{r}_{2})\cdot(\boldsymbol{r}_{3}\times\boldsymbol{r}_{4})],\\
         8\cdot\operatorname{Tr}(\rho_{1}\rho_{3}\rho_{4}\rho_{2}) &=
         1+\sum\limits_{i<j}\boldsymbol{r}_{i}\cdot\boldsymbol{r}_{j}+(\boldsymbol{r}_{1}\cdot\boldsymbol{r}_{3})(\boldsymbol{r}_{2}\cdot\boldsymbol{r}_{4})
         -(\boldsymbol{r}_{1}\cdot\boldsymbol{r}_{4})(\boldsymbol{r}_{2}\cdot\boldsymbol{r}_{3})+(\boldsymbol{r}_{1}\cdot\boldsymbol{r}_{2})(\boldsymbol{r}_{3}\cdot\boldsymbol{r}_{4})\nonumber\\
         &\quad +i[(\boldsymbol{r}_{4}+\boldsymbol{r}_{2})\cdot(\boldsymbol{r}_{1}\times\boldsymbol{r}_{3})+(\boldsymbol{r}_{1}+\boldsymbol{r}_{3})\cdot(\boldsymbol{r}_{4}\times\boldsymbol{r}_{2})],\\
         8\cdot\operatorname{Tr}(\rho_{1}\rho_{4}\rho_{2}\rho_{3}) &=
         1+\sum\limits_{i<j}\boldsymbol{r}_{i}\cdot\boldsymbol{r}_{j}+(\boldsymbol{r}_{1}\cdot\boldsymbol{r}_{4})(\boldsymbol{r}_{2}\cdot\boldsymbol{r}_{3})
         -(\boldsymbol{r}_{1}\cdot\boldsymbol{r}_{2})(\boldsymbol{r}_{3}\cdot\boldsymbol{r}_{4})+(\boldsymbol{r}_{1}\cdot\boldsymbol{r}_{3})(\boldsymbol{r}_{2}\cdot\boldsymbol{r}_{4})\nonumber\\
         &\quad +i[(\boldsymbol{r}_{2}+\boldsymbol{r}_{3})\cdot(\boldsymbol{r}_{1}\times\boldsymbol{r}_{4})+(\boldsymbol{r}_{1}+\boldsymbol{r}_{4})\cdot(\boldsymbol{r}_{2}\times\boldsymbol{r}_{3})].
    \end{flalign*}
    With help of this and \eqref{keyeq-4imaginarity}, we obtain
    \begin{flalign}
        1+\sum\limits_{i<j}\boldsymbol{r}_{i}\cdot\boldsymbol{r}_{j}+(\boldsymbol{r}_{1}\cdot\boldsymbol{r}_{2})(\boldsymbol{r}_{3}\cdot\boldsymbol{r}_{4})
         -(\boldsymbol{r}_{1}\cdot\boldsymbol{r}_{3})(\boldsymbol{r}_{2}\cdot\boldsymbol{r}_{4})+(\boldsymbol{r}_{1}\cdot\boldsymbol{r}_{4})(\boldsymbol{r}_{2}\cdot\boldsymbol{r}_{3})=0,\label{eq2.1-4imaginarity}\\
         1+\sum\limits_{i<j}\boldsymbol{r}_{i}\cdot\boldsymbol{r}_{j}+(\boldsymbol{r}_{1}\cdot\boldsymbol{r}_{3})(\boldsymbol{r}_{2}\cdot\boldsymbol{r}_{4})
         -(\boldsymbol{r}_{1}\cdot\boldsymbol{r}_{4})(\boldsymbol{r}_{2}\cdot\boldsymbol{r}_{3})+(\boldsymbol{r}_{1}\cdot\boldsymbol{r}_{2})(\boldsymbol{r}_{3}\cdot\boldsymbol{r}_{4})=0,\label{eq2.2-4imaginarity}\\
         1+\sum\limits_{i<j}\boldsymbol{r}_{i}\cdot\boldsymbol{r}_{j}+(\boldsymbol{r}_{1}\cdot\boldsymbol{r}_{4})(\boldsymbol{r}_{2}\cdot\boldsymbol{r}_{3})
         -(\boldsymbol{r}_{1}\cdot\boldsymbol{r}_{2})(\boldsymbol{r}_{3}\cdot\boldsymbol{r}_{4})+(\boldsymbol{r}_{1}\cdot\boldsymbol{r}_{3})(\boldsymbol{r}_{2}\cdot\boldsymbol{r}_{4})=0.\label{eq2.3-4imaginarity}
    \end{flalign}
    By calculating \eqref{eq2.1-4imaginarity}$-$\eqref{eq2.2-4imaginarity} and \eqref{eq2.1-4imaginarity}$-$\eqref{eq2.3-4imaginarity}, we have
    \begin{equation*}
        \begin{cases}
            2(\boldsymbol{r}_{1}\cdot\boldsymbol{r}_{4})(\boldsymbol{r}_{2}\cdot\boldsymbol{r}_{3})-2(\boldsymbol{r}_{1}\cdot\boldsymbol{r}_{3})(\boldsymbol{r}_{2}\cdot\boldsymbol{r}_{4})=0,\\
            2(\boldsymbol{r}_{1}\cdot\boldsymbol{r}_{2})(\boldsymbol{r}_{3}\cdot\boldsymbol{r}_{4})
         -2(\boldsymbol{r}_{1}\cdot\boldsymbol{r}_{3})(\boldsymbol{r}_{2}\cdot\boldsymbol{r}_{4})=0.
        \end{cases}
    \end{equation*}
    Hence $(\boldsymbol{r}_{1}\cdot\boldsymbol{r}_{2})(\boldsymbol{r}_{3}\cdot\boldsymbol{r}_{4})=(\boldsymbol{r}_{1}\cdot\boldsymbol{r}_{3})(\boldsymbol{r}_{2}\cdot\boldsymbol{r}_{4})=(\boldsymbol{r}_{1}\cdot\boldsymbol{r}_{4})(\boldsymbol{r}_{2}\cdot\boldsymbol{r}_{3})$, which implies that
    \begin{equation}\label{eq3.1-4imaginarity}
        (\boldsymbol{u}_{1}\cdot\boldsymbol{u}_{2})(\boldsymbol{u}_{3}\cdot\boldsymbol{u}_{4})=(\boldsymbol{u}_{1}\cdot\boldsymbol{u}_{3})(\boldsymbol{u}_{2}\cdot\boldsymbol{u}_{4})=(\boldsymbol{u}_{1}\cdot\boldsymbol{u}_{4})(\boldsymbol{u}_{2}\cdot\boldsymbol{u}_{3}).
    \end{equation}
    As is well known, we have the following fact
    \begin{equation}\label{eq3.2-4imaginarity}
        \det(\boldsymbol{a},\boldsymbol{b})\det(\boldsymbol{c},\boldsymbol{d})=(\boldsymbol{a}\cdot\boldsymbol{c})(\boldsymbol{b}\cdot\boldsymbol{d})-(\boldsymbol{a}\cdot\boldsymbol{d})(\boldsymbol{b}\cdot\boldsymbol{c})
    \end{equation}
    for $\forall\boldsymbol{a},\boldsymbol{b},\boldsymbol{c},\boldsymbol{d}\in\mathbb{R}^{2}$. From \eqref{eq3.1-4imaginarity} and \eqref{eq3.2-4imaginarity}, we can see that 
    \begin{equation}
    \begin{cases}
    \det(\boldsymbol{u}_{1},\boldsymbol{u}_{2})
    \det(\boldsymbol{u}_{3},\boldsymbol{u}_{4})
    =(\boldsymbol{u}_{1}\cdot\boldsymbol{u}_{3})(\boldsymbol{u}_{2}\cdot\boldsymbol{u}_{4})-(\boldsymbol{u}_{1}\cdot\boldsymbol{u}_{4})(\boldsymbol{u}_{2}\cdot\boldsymbol{u}_{3})=0,\\
    \det(\boldsymbol{u}_{1},\boldsymbol{u}_{3})
    \det(\boldsymbol{u}_{2},\boldsymbol{u}_{4})
    =(\boldsymbol{u}_{1}\cdot\boldsymbol{u}_{2})(\boldsymbol{u}_{3}\cdot\boldsymbol{u}_{4})-(\boldsymbol{u}_{1}\cdot\boldsymbol{u}_{4})(\boldsymbol{u}_{2}\cdot\boldsymbol{u}_{3})=0,\\
    \det(\boldsymbol{u}_{1},\boldsymbol{u}_{4})
    \det(\boldsymbol{u}_{2},\boldsymbol{u}_{3})
    =(\boldsymbol{u}_{1}\cdot\boldsymbol{u}_{2})(\boldsymbol{u}_{3}\cdot\boldsymbol{u}_{4})-(\boldsymbol{u}_{1}\cdot\boldsymbol{u}_{3})(\boldsymbol{u}_{2}\cdot\boldsymbol{u}_{4})=0.
    \end{cases}
    \end{equation}
    With help of this and $\boldsymbol{u}_{k}\neq 0$ for all $k\in\{1,2,3,4\}$, we obtain that at least three vectors in the sequence $\{\boldsymbol{u}_{k}\}_{k=1}^{4}$ are collinear, which is equivalent to at least three vectors in the sequence $\{\boldsymbol{r}_{k}\}_{k=1}^{4}$ are collinear. Therefore, there exists some $s\in\{1,2,3,4\}$ such that $[\rho_{k},\rho_{l}]=0$ for all $k,l\in\{1,2,3,4\}\backslash\{s\}$. From Proposition \ref{pp:rho1rho2rho3rho4}, we can finish the proof. 
\end{proof}

     Proposition \ref{4imaginarity} can be used to verify multistate imaginarity. To illustrate this, consider the states $\rho_{1},\rho_{2},\rho_{3},\rho_{4}$ from Example \ref{example1}. Using Proposition \ref{4imaginarity}, one can readily verify that the tuple $\vec{\rho}=(\rho_{1},\rho_{2},\rho_{3},\rho_{4})$ has imaginarity. To show that the condition $\operatorname{Tr}(\rho_{k}^{2})>\frac{1}{2}$ for all $k\in\{1,2,3,4\}$ is necessary in Proposition \ref{4imaginarity}, we provide the following example.

\begin{example}
     Consider the following four density operators on $\mathbb{C}^{2}$:
    \begin{equation*}
    \begin{split}
    \rho_{1}=\begin{pmatrix} \frac{1}{2} & 0\\ 0 & \frac{1}{2} \end{pmatrix} &, 
    \rho_{2}=\begin{pmatrix} \frac{1}{2} & \frac{1}{\sqrt{6}}\\ \frac{1}{\sqrt{6}} & \frac{1}{2} \end{pmatrix}, \\
    \rho_{3}=\frac{1}{2}\begin{pmatrix} 1+\frac{1}{\sqrt{2}} & -\frac{1}{\sqrt{6}}\\ -\frac{1}{\sqrt{6}} & 1-\frac{1}{\sqrt{2}}\end{pmatrix} &,
    \rho_{4}=\frac{1}{2}\begin{pmatrix} 1-\frac{1}{\sqrt{2}} & -\frac{1}{\sqrt{6}}\\ -\frac{1}{\sqrt{6}} & 1+\frac{1}{\sqrt{2}}\end{pmatrix}.
    \end{split}
    \end{equation*}
    By direct computation, we have $\operatorname{Tr}(\rho_{1}\rho_{2}\rho_{3}\rho_{4})=\operatorname{Tr}(\rho_{1}\rho_{3}\rho_{4}\rho_{2})=\operatorname{Tr}(\rho_{1}\rho_{4}\rho_{2}\rho_{3})=0$, and $\operatorname{Tr}(\rho_{i}\rho_{j})>0$ for all distinct $i,j\in\{1,2,3,4\}$. Clearly, $\vec{\rho}=(\rho_{1},\rho_{2},\rho_{3},\rho_{4})$ is imaginarity-free. Hence this example shows that the condition $\operatorname{Tr}(\rho_{k}^{2})>\frac{1}{2}$ for all $k\in\{1,2,3,4\}$ in Proposition \ref{4imaginarity} is indispensable. 
\end{example}

    Now, we propose the following theorem for the general case $n\geq 4$.

\begin{theorem}\label{nimaginarity}
    Let $\vec{\rho}=(\rho_{1},\cdots,\rho_{n})$ be an $n$-tuple of single qubit states where $\rho_{i}\in \mathcal{D}(\mathbb{C}^{2})$ and $n\geq 4$. Denote the Bargmann scenario $\mathcal{W}=\{(\tau(1),\tau(2),\cdots,\tau(n))|\tau\in S_{n}\enspace\mbox{and}\enspace\tau(1)=1\}$, where $S_{n}$ is the symmetric group on $\{1,2,\cdots,n\}$. Suppose that 
    \begin{equation}\label{keyeq-nimaginarity}
       \Delta_{w}(\vec{\rho})=0,\enspace\forall w\in\mathcal{W}.
    \end{equation}
    and $\operatorname{Tr}(\rho_{k}^{2})>\frac{1}{2}$ for all $k\in\{1,2,\cdots,n\}$. Then the following statements are equivalent:\\
    (\rmnum{1}) There exist some $i,j\in\{1,2,\cdots,n\}$ such that $\operatorname{Tr}(\rho_{i}\rho_{j})=0$.\\
    (\rmnum{2}) The tuple $\vec{\rho}=(\rho_{1},\rho_{2},\cdots,\rho_{n})$ is imaginarity-free.
\end{theorem}
\begin{proof}
    We prove Theorem \ref{nimaginarity} in Appendix \ref{Appendix nimaginarity}. 
\end{proof}

According to Proposition \ref{pp:rho1rho2rho3rho4}, Proposition \ref{thm_5rho}, Theorem \ref{thm_nrho} and Theorem \ref{nimaginarity}, we immediately obtain the following corollary.

\begin{corollary}\label{maincor}
    Let $\vec{\rho}=(\rho_{1},\cdots,\rho_{n})$ be an $n$-tuple of single qubit states where $\rho_{i}\in \mathcal{D}(\mathbb{C}^{2})$ and $n\geq 4$. Denote the Bargmann scenario $\mathcal{W}=\{(\tau(1),\tau(2),\cdots,\tau(n))|\tau\in S_{n}\enspace\mbox{and}\enspace\tau(1)=1\}$, where $S_{n}$ is the symmetric group on $\{1,2,\cdots,n\}$. Suppose that 
    \begin{equation*}
       \Delta_{w}(\vec{\rho})=0,\enspace\forall w\in\mathcal{W}.
    \end{equation*}
    and $\operatorname{Tr}(\rho_{k}^{2})>\frac{1}{2}$ for all $k\in\{1,2,\cdots,n\}$. Then the following statements are equivalent:\\
    (\rmnum{1}) The tuple $\vec{\rho}=(\rho_{1},\rho_{2},\cdots,\rho_{n})$ has imaginarity.\\
    (\rmnum{2}) For any $s\in\{1,2,\cdots,n\}$, the tuple $(\rho_{k})_{k\neq s}$ is coherent.\\
    (\rmnum{3}) For any $i,j\in\{1,2,\cdots,n\}$, ${\rm Tr}(\rho_{i}\rho_{j})>0$.
\end{corollary}

This result depicts the connection between set imaginarity and set coherence of qubit states under vanishing Bargmann invariants.

\section{Conclusions and discussions}\label{sec5}
We have studied the connection between set coherence and set imaginarity for an $n$-tuple of qubit states when the Bargmann invariants of order $n$ vanish, where $n\geq 4$. Actually, we discussed the relationship between set coherence, set imaginarity, and orthogonal states. On the one hand, we show that if some Bargmann invariants of order $n$ vanish, then the existence of a pair of orthogonal states is equivalent to the existence of a state whose removal makes the tuple of remaining states incoherent. On the other hand, assuming that the tuple does not contain the maximally mixed state, we prove that the same vanishing conditions imply that a tuple is imaginarity-free if and only if it contains a pair of orthogonal states. According to these results, we obtain that, when some higher-order Bargmann invariants vanish, a tuple of single-qubit states has imaginarity if and only if every tuple obtained by deleting one of its states is coherent. These results can also be used to witness set coherence and set imaginarity of qubit states under specific conditions. Since Bargmann invariants can be estimated directly, this method can be experimentally accessed without state tomography.\par 

We expect that the vanishing conditions of Theorem \ref{thm_nrho} and Theorem \ref{nimaginarity} can be further weakened. In view of Propositions \ref{pp:rho1rho2rho3rho4}, \ref{thm_5rho} and \ref{4imaginarity}, we conjecture that it may be possible to obtain a feasible weakening of these vanishing conditions based on the Bargmann scenario associated with the alternating group $A_{n}$. In addition, it would be worthwhile to explore whether higher-order Bargmann invariants can lead to an analogous equivalence between set coherence and set imaginarity in high-dimensional systems. We believe that such research can help us better understand set coherence and set imaginarity.\par

\bigskip 
    \noindent{\bf Acknowledgments} 
    This work is supported by the National Natural Science Foundation of China (NSFC) under Grant No.12471427. The authors used Kimi K3 to improve the proof of Theorem \ref{thm_nrho}, and to assist with language polishing. All mathematical statements and proofs were independently verified by the authors, who take responsibility for the final manuscript.


\begin{appendices}
\section{The proof of Lemma \ref{pp:rho sigma}}\label{Appendix pp:rho sigma}
{\bf The proof of Lemma \ref{pp:rho sigma}}: Choosing a basis $\ket{0},\ket{1}$ in which 
\begin{equation}\label{eq1 for pp:rho sigma}
    \rho=\begin{pmatrix} p & 0 \\ 0 & 1-p \end{pmatrix}, \sigma=\begin{pmatrix} q & \beta \\ \overline{\beta} & 1-q \end{pmatrix},
\end{equation}
where $p,q\in [0,1]$, $\beta\in\mathbb{C}$ and $|\beta|^{2}\leq q(1-q)$. Assume that $w=(l_{1},\cdots,l_n)$ and $\Delta_{w}(\rho,\sigma)=\operatorname{Tr}(\rho_{l_{1}}\cdots\rho_{l_{n}})$, where $\rho_{l_{k}}\in\{\rho,\sigma\}$. Note that 
    \begin{equation}\label{eq2 for pp:rho sigma}
    \Delta_{w}(\rho,\sigma)=\operatorname{Tr}(\rho_{l_{1}}\cdots\rho_{l_{n}})=\sum\limits_{j_{1},j_{2},\cdots,j_{n}\in\{0,1\}}(\rho_{l_{1}})_{j_{1},j_{2}}(\rho_{l_{2}})_{j_{2},j_{3}}\cdots(\rho_{l_{n}})_{j_{n},j_{1}}.
    \end{equation}
Eq.(\ref{eq1 for pp:rho sigma}) implies that $(\rho)_{01}=(\rho)_{10}=0$ and $(\sigma)_{01}=\overline{(\sigma)_{10}}$. Then
$(\rho_{l_{1}})_{j_{1},j_{2}}(\rho_{l_{2}})_{j_{2},j_{3}}\cdots(\rho_{l_{n}})_{j_{n},j_{1}}\geq 0$ for all $j_{1},j_{2},\cdots,j_{n}\in\{0,1\}$, as $01$ and $10$ appear in pairs. Hence $\Delta_{w}(\rho,\sigma)\geq 0$. \par
Suppose that $\rho\sigma\neq 0$. We show that $\Delta_{w}(\rho,\sigma)\neq 0$ for all $w\in\mathcal{W}$. Since $\rho,\sigma$ are density operators, we have $\operatorname{Tr}(\rho^{n})>0$ and $\operatorname{Tr}(\sigma^{n})>0$ for all $n\in\mathbb{Z}^{+}$, which implies that $\Delta_{w}(\rho,\sigma)>0$ when $w$ consists of a single repeated letter. Consider the case that $w=(l_{1},\cdots,l_n)$ contains two different letters. We claim that there exists $i\in\{0,1\}$ such that $\braket{i|\rho|i}>0$ and $\braket{i|\sigma|i}>0$. If $q\in (0,1)$, we just need to take $i=0$ when $p\neq 0$ and $i=1$ when $p=0$. If $q=0$ or $q=1$, then $\beta=0$ from $|\beta|^{2}\leq q(1-q)$. From $\rho\sigma\neq 0$, we just need to take $i=1$ when $q=0$ and take $i=0$ when $q=1$. Hence the claim is proved. Therefore, we have $(\rho_{l_{1}})_{ii}\cdots(\rho_{l_{n}})_{ii}>0$. Then $\Delta_{w}(\rho,\sigma)>0$ from Eq.(\ref{eq2 for pp:rho sigma}) and $(\rho_{l_{1}})_{j_{1},j_{2}}(\rho_{l_{2}})_{j_{2},j_{3}}\cdots(\rho_{l_{n}})_{j_{n},j_{1}}\geq 0$ for all $j_{1},j_{2},\cdots,j_{n}\in\{0,1\}$. On the contrary, we can see that $\Delta_{12}(\rho,\sigma)=\operatorname{Tr}(\rho\sigma)=0$ when $\rho\sigma=0$. \qed

\section{The proof of Theorem \ref{thm_nrho}}\label{Appendix thm_nrho}
To prove Theorem \ref{thm_nrho} using induction method, we need to first prove the following lemma.
\begin{lemma}\label{thm_6rho}
    Let $\rho_{1},\rho_{2},\rho_{3},\rho_{4},\rho_{5},\rho_{6}$ be six density operators on $\mathbb{C}^{2}$. Denote the Bargmann scenario $\mathcal{W}=\{(\sigma(1),\sigma(2),\sigma(3),\sigma(4),\sigma(5),\sigma(6))|\sigma\in S_{6}\enspace\mbox{and}\enspace\sigma(1)=1\}$, where $S_{6}$ is the symmetric group on $\{1,2,3,4,5,6\}$. Suppose that 
    \begin{equation}\label{key condition of pp:6qubit}
        \Delta_{w}(\rho_{1},\rho_{2},\rho_{3},\rho_{4},\rho_{5},\rho_{6})=0,\enspace \forall w\in\mathcal{W}.
    \end{equation}
    If there exists some $i,j\in\{1,2,3,4,5,6\}$ such that $\rho_{i}\rho_{j}=0$, then there exists some $s\in\{1,2,3,4,5,6\}$ such that $[\rho_{k},\rho_{l}]=0$ for all $k,l\in\{1,2,3,4,5,6\}\backslash\{s\}$.
\end{lemma}
\begin{proof}
    Assume that $\{1,2,3,4,5,6\}\backslash\{i,j\}=\{k_{1},k_{2},k_{3},k_{4}\}$. Due to $\rho_{i}\rho_{j}=0$, we can choose a basis $\ket{0},\ket{1}$ in which 
    \begin{equation}
    \rho_{i}=\begin{pmatrix} 1 & 0 \\ 0 & 0 \end{pmatrix},
    \rho_{j}=\begin{pmatrix} 0 & 0 \\ 0 & 1 \end{pmatrix},
    \rho_{k_{l}}=\begin{pmatrix} p_{l} & \alpha_{l}\\ \overline{\alpha_{l}} & 1-p_{l} \end{pmatrix}, l=1,2,3,4
    \end{equation}
where $p_{l}\in [0,1]$, $\alpha_{l}\in \mathbb{C}$ and $|\alpha_{l}|^{2}\leq p_{l}(1-p_{l})$. Clearly, $p_{l}\in (0,1)$ when $\alpha_{l}\neq 0$. Due to cyclicity of the trace, we see that \eqref{key condition of pp:6qubit} is equivalent to ${\rm Tr}(\rho_{\sigma(1)}\rho_{\sigma(2)}\rho_{\sigma(3)}\rho_{\sigma(4)}\rho_{\sigma(5)}\rho_{\sigma(6)})=0$ for all $\sigma\in S_{6}$. We define 
\begin{eqnarray}\label{eq:X st}    
X_{st}=\overline{\alpha_{s}}p_{t}+(1-p_{s})\overline{\alpha_{t}}.
\end{eqnarray}
Then we have
\begin{equation}\label{keyeq:0=Tr rho6}
\begin{cases}
0={\rm Tr}(\rho_{i}\rho_{k_{1}}\rho_{k_{2}}\rho_{j}\rho_{k_{4}}\rho_{k_{3}})=\overline{X_{21}}X_{43},\\
0={\rm Tr}(\rho_{i}\rho_{k_{2}}\rho_{k_{1}}\rho_{j}\rho_{k_{3}}\rho_{k_{4}})=\overline{X_{12}}X_{34},\\
0={\rm Tr}(\rho_{i}\rho_{k_{2}}\rho_{k_{3}}\rho_{j}\rho_{k_{4}}\rho_{k_{1}})=\overline{X_{32}}X_{41},\\
0={\rm Tr}(\rho_{i}\rho_{k_{3}}\rho_{k_{1}}\rho_{j}\rho_{k_{4}}\rho_{k_{2}})=\overline{X_{13}}X_{42},\\
0={\rm Tr}(\rho_{i}\rho_{k_{4}}\rho_{k_{1}}\rho_{j}\rho_{k_{2}}\rho_{k_{3}})=\overline{X_{14}}X_{23},\\
0={\rm Tr}(\rho_{i}\rho_{k_{4}}\rho_{k_{2}}\rho_{j}\rho_{k_{3}}\rho_{k_{1}})=\overline{X_{24}}X_{31},\\
0={\rm Tr}(\rho_{i}\rho_{k_{1}}\rho_{j}\rho_{k_{2}}\rho_{k_{3}}\rho_{k_{4}})=\alpha_{1}[\overline{\alpha_{2}}(p_{3}p_{4}+\alpha_{3}\overline{\alpha_{4}})+(1-p_{2})X_{34}],\\
0={\rm Tr}(\rho_{i}\rho_{k_{2}}\rho_{j}\rho_{k_{1}}\rho_{k_{4}}\rho_{k_{3}})=\alpha_{2}[\overline{\alpha_{1}}(p_{4}p_{3}+\alpha_{4}\overline{\alpha_{3}})+(1-p_{1})X_{43}],\\
0={\rm Tr}(\rho_{i}\rho_{k_{3}}\rho_{j}\rho_{k_{2}}\rho_{k_{4}}\rho_{k_{1}})=\alpha_{3}[\overline{\alpha_{2}}(p_{4}p_{1}+\alpha_{4}\overline{\alpha_{1}})+(1-p_{2})X_{41}],\\
0={\rm Tr}(\rho_{i}\rho_{k_{4}}\rho_{j}\rho_{k_{2}}\rho_{k_{1}}\rho_{k_{3}})=\alpha_{4}[\overline{\alpha_{2}}(p_{1}p_{3}+\alpha_{1}\overline{\alpha_{3}})+(1-p_{2})X_{13}].\\
\end{cases}
\end{equation}
We next show that at most one of $\alpha_{1},\alpha_{2},\alpha_{3},\alpha_{4}$ is non-zero. For this purpose, three cases \Rmnum{1}-\Rmnum{3} are distinguished according to how many of $\alpha_{1},\alpha_{2},\alpha_{3},\alpha_{4}$ are non-zero. In case \Rmnum{1}, we will define a set $U$ and obtain contradictions by studying its structure.\par

\noindent
{\bf Case \Rmnum{1}:} Consider the case where $\alpha_{l}\neq 0$ for all $l\in\{1,2,3,4\}$. We define 
\begin{equation}\label{def_U}
    U=\{\{s,t\}|\overline{X_{st}}=0\enspace\mbox{or}\enspace\overline{X_{ts}}=0\}
\end{equation}
by \eqref{eq:X st}. For any $\{s,t\}\in U$, we have $\frac{\alpha_{s}}{\alpha_{t}}\in (-\infty,0)$ because $\frac{\alpha_{s}}{\alpha_{t}}=-\frac{1-p_{s}}{p_{t}}$ when $\overline{X_{st}}=0$ and $\frac{\alpha_{s}}{\alpha_{t}}=-\frac{p_{s}}{1-p_{t}}$ when $\overline{X_{ts}}=0$.
We claim that there are no three distinct numbers $a,b,c\in\{1,2,3,4\}$ such that $\{\{a,b\},\{b,c\},\{c,a\}\}\subseteq U$. Indeed, if such $a,b,c$ existed, then $\frac{\alpha_{a}}{\alpha_{b}},\frac{\alpha_{b}}{\alpha_{c}},\frac{\alpha_{c}}{\alpha_{a}}\in (-\infty,0)$, which would force their product to be negative, contradicting the fact that $\frac{\alpha_{a}}{\alpha_{b}}\frac{\alpha_{b}}{\alpha_{c}}\frac{\alpha_{c}}{\alpha_{a}}=1>0$. Thus the claim is proved. Let $M_{1}=\{\{1,2\},\{3,4\}\}$, $M_{2}=\{\{1,4\},\{2,3\}\}$, $M_{3}=\{\{1,3\},\{2,4\}\}$. From \eqref{keyeq:0=Tr rho6}, we can see that 
\begin{eqnarray}\label{eq:U cap Mr}    
U\cap M_{r}\neq\varnothing
\end{eqnarray}
for all $r\in\{1,2,3\}$. 
If $M_{n}\subseteq U$ for some $n\in\{1,2,3\}$, then there would exist three distinct numbers $a,b,c\in\{1,2,3,4\}$ such that $\{\{a,b\},\{b,c\},\{c,a\}\}\subseteq U$, contradicting the previous claim. 
Hence $M_{r}\not\subseteq U$  for all $r\in\{1,2,3\}$. 
Combining with \eqref{eq:U cap Mr} and the claim above \eqref{eq:U cap Mr}, 
we obtain that there exists some $a\in\{1,2,3,4\}$ such that $U=\{\{a,x\}|x\in\{1,2,3,4\}\backslash\{a\}\}$. Assume without loss of generality that $U=\{\{1,2\},\{1,3\},\{1,4\}\}$. Combining this with \eqref{def_U}, we see that $\overline{X_{34}}, \overline{X_{43}},\overline{X_{23}},\overline{X_{32}},\overline{X_{24}},\overline{X_{42}}$ are all non-zero. Then according to \eqref{keyeq:0=Tr rho6}, we obtain
\begin{equation*}
    \begin{cases}
    \alpha_{1}p_{t}+(1-p_{1})\alpha_{t}=\overline{X_{1t}}=0,\\
    \alpha_{t}p_{1}+(1-p_{t})\alpha_{1}=\overline{X_{t1}}=0,\\
    \end{cases}
\end{equation*}
for all $t\in\{2,3,4\}$. So we have $\alpha_{1}p_{t}-\frac{(1-p_{1})(1-p_{t})}{p_{1}}\alpha_{1}=0$ (Note that $p_1\neq 0$ from $\alpha_{1}\neq 0$), which is equivalent to $p_{1}p_{t}=1-p_{1}-p_{t}+p_{1}p_{t}$. Consequently, $p_{t}=1-p_{1}$ and $\alpha_{t}=-\alpha_{1}$ for all $t\in\{2,3,4\}$. Using 
\eqref{eq:rho_i}, we have 
\begin{equation*}
    \begin{split}
    {\rm Tr}(\rho_{i}\rho_{k_{1}}\rho_{j}\rho_{k_{2}}\rho_{k_{3}}\rho_{k_{4}})
    & =\alpha_{1}[-\overline{\alpha_{1}}(1-2p_{1}+p_{1}^{2}+|\alpha_{1}|^{2})-p_{1}\overline{\alpha_{1}}]\\
    & =\alpha_{1}[-\overline{\alpha_{1}}+p_{1}\overline{\alpha_{1}}-p_{1}^{2}\overline{\alpha_{1}}-|\alpha_{1}|^{2}\overline{\alpha_{1}}]\\
    & =-[(1-p_{1}+p_{1}^{2})|\alpha_{1}|^{2}+|\alpha_{1}|^{4}]<0, 
    \end{split}  
\end{equation*}
contradicting ${\rm Tr}(\rho_{i}\rho_{k_{1}}\rho_{j}\rho_{k_{2}}\rho_{k_{3}}\rho_{k_{4}})=0$ in \eqref{keyeq:0=Tr rho6}.\par

\noindent
{\bf Case \Rmnum{2}:} We consider the case that three elements of $\{\alpha_{1},\alpha_{2},\alpha_{3},\alpha_{4}\}$ are non-zero. Assume without loss of generality that $\alpha_{1}\neq 0$, $\alpha_{2}\neq 0$, $\alpha_{3}\neq 0$, and $\alpha_{4}=0$. We claim that $0=\overline{X_{12}}=\overline{X_{23}}=\overline{X_{31}}$ or $0=\overline{X_{21}}=\overline{X_{32}}=\overline{X_{13}}$. Due to $p_{4}\in [0,1]$, the proof of this claim can split into three cases. If $p_{4}\in (0,1)$, then according to $\alpha_{4}=0$ and \eqref{eq:X st}, we have $X_{t4}=\overline{\alpha_{t}}p_{4}\neq 0$ and $X_{4t}=(1-p_{4})\overline{\alpha_{t}}\neq 0$ for all $t\in\{1,2,3\}$. Combining this with \eqref{keyeq:0=Tr rho6}, we obtain $0=\overline{X_{12}}=\overline{X_{23}}=\overline{X_{31}}$ and $0=\overline{X_{21}}=\overline{X_{32}}=\overline{X_{13}}$. If $p_{4}=0$, then $X_{4t}=\overline{\alpha_{t}}\neq 0$. Combining this with \eqref{keyeq:0=Tr rho6}, we obtain $0=\overline{X_{12}}=\overline{X_{23}}=\overline{X_{31}}$. If $p_{4}=1$, then $X_{t4}=\overline{\alpha_{s}}\neq 0$. Combining this with \eqref{keyeq:0=Tr rho6}, we obtain $0=\overline{X_{21}}=\overline{X_{32}}=\overline{X_{13}}$. Hence the claim is proved. Therefore, according to \eqref{eq:X st}, we have $\frac{\alpha_{1}}{\alpha_{2}},\frac{\alpha_{2}}{\alpha_{3}},\frac{\alpha_{3}}{\alpha_{1}}\in (-\infty,0)$ (Note that $p_{t}\in (0,1)$ from $\alpha_{t}\neq 0$), which would force their product to be negative, contradicting the fact that $\frac{\alpha_{1}}{\alpha_{2}}\frac{\alpha_{2}}{\alpha_{3}}\frac{\alpha_{3}}{\alpha_{1}}=1>0$.\par

\noindent
{\bf Case \Rmnum{3}:} 
We consider the case that two elements of $\{\alpha_{1},\alpha_{2},\alpha_{3},\alpha_{4}\}$ are non-zero. Assume without loss of generality that $\alpha_{1}\neq 0$, $\alpha_{2}\neq 0$, $\alpha_{3}=0$, and $\alpha_{4}=0$. Then we can see that
\begin{equation*}
\begin{cases}
0={\rm Tr}(\rho_{i}\rho_{k_{2}}\rho_{k_{3}}\rho_{j}\rho_{k_{4}}\rho_{k_{1}})=\overline{X_{32}}X_{41}=(1-p_{3})\alpha_{2}(1-p_{4})\overline{\alpha_{1}},\\
0={\rm Tr}(\rho_{i}\rho_{k_{3}}\rho_{k_{1}}\rho_{j}\rho_{k_{4}}\rho_{k_{2}})=\overline{X_{13}}X_{42}=\alpha_{1}p_{3}(1-p_{4})\overline{\alpha_{2}},\\
0={\rm Tr}(\rho_{i}\rho_{k_{4}}\rho_{k_{1}}\rho_{j}\rho_{k_{2}}\rho_{k_{3}})=\overline{X_{14}}X_{23}=\alpha_{1}p_{4}\overline{\alpha_{2}}p_{3},\\
0={\rm Tr}(\rho_{i}\rho_{k_{4}}\rho_{k_{2}}\rho_{j}\rho_{k_{3}}\rho_{k_{1}})=\overline{X_{24}}X_{31}=\alpha_{2}p_{4}(1-p_{3})\overline{\alpha_{1}},\\
\end{cases}
\end{equation*}
Therefore, we have
\begin{equation*}
\begin{cases}
0=(1-p_{3})(1-p_{4}),\\
0=p_{3}(1-p_{4}),\\
0=p_{4}p_{3},\\
0=p_{4}(1-p_{3}),
\end{cases}
\end{equation*}
which is impossible. \par

Based on the above cases \Rmnum{1}-\Rmnum{3}, we obtain that there exist some $s_{1},s_{2},s_{3}\in\{1,2,3,4\}$ satisfying $\alpha_{s_{1}}=\alpha_{s_{2}}=\alpha_{s_{3}}=0$. As a result, $[\rho_{k},\rho_{l}]=0$ for all $k,l\in\{i,j,k_{s_{1}},k_{s_{2}},k_{s_{3}}\}$ by \eqref{eq:rho_i}.
\end{proof}

    Now, we can prove Theorem \ref{thm_nrho} using induction method.

{\bf The proof of Theorem \ref{thm_nrho}}: (\rmnum{1})$\Rightarrow$(\rmnum{2}). Assume that $\{1,2,\cdots,n\}\backslash\{i,j\}=\{k_{1},k_{2},\cdots,k_{n-2}\}$. Due to $\rho_{i}\rho_{j}=0$, we can choose a basis $\ket{0},\ket{1}$ in which 
    \begin{equation}\label{eq:nrho_i}
    \rho_{i}=\begin{pmatrix} 1 & 0 \\ 0 & 0 \end{pmatrix},
    \rho_{j}=\begin{pmatrix} 0 & 0 \\ 0 & 1 \end{pmatrix},
    \rho_{k_{l}}=\begin{pmatrix} p_{l} & \alpha_{l}\\ \overline{\alpha_{l}} & 1-p_{l} \end{pmatrix}, l=1,2,\cdots,n-2
    \end{equation}
    where $p_{l}\in [0,1]$, $\alpha_{l}\in \mathbb{C}$ and $|\alpha_{l}|^{2}\leq p_{l}(1-p_{l})$. Due to cyclicity of the trace, we see that \eqref{key condition of pp:nqubit} is equivalent to ${\rm Tr}(\rho_{\sigma(1)}\rho_{\sigma(2)}\cdots\rho_{\sigma(n)})=0$ for all $\sigma\in S_{n}$.\par
    In Lemma \ref{thm_6rho}, we have proven the claim (\rmnum{1}) $\Rightarrow$ (\rmnum{2}) for $n=6$. To apply the induction on $n$, we assume that the claim is true when $n=m-1$. We need prove the claim for $n=m$. For any $l\in\{1,2,\cdots,m-2\}$ and $A,B\in\mathcal{M}_{2}(\mathbb{C})$, we have
    \begin{equation}
    \label{eq:Tr rhoi}
        \begin{split}
        {\rm Tr}(\rho_{i}A\rho_{j}\rho_{k_{l}}B) &={\rm Tr}(\ket{0}\bra{0}A\ket{1}\bra{1}\rho_{k_{l}}B)=\braket{0|A|1}{\rm Tr}(\ket{0}\bra{1}\rho_{k_{l}}B)\\
        & =A_{01}\braket{1|\rho_{k_{l}}B|0}=\overline{\alpha_{l}}A_{01}B_{00}+A_{01}(1-p_{l})B_{10}\\
        & =\overline{\alpha_{l}}A_{01}{\rm Tr}(\rho_{i}B)+(1-p_{l})A_{01}\braket{1|B|0}\\
        & =\overline{\alpha_{l}}A_{01}{\rm Tr}(\rho_{i}B)+(1-p_{l}){\rm Tr}(\rho_{i}A\rho_{j}B).
        \end{split}
    \end{equation}
    To prove (\rmnum{1})$\Rightarrow$(\rmnum{2}), we distinguish three cases \Rmnum{1}–\Rmnum{3} according to whether all $\alpha_{l}$ in \eqref{eq:nrho_i} are non-zero. Specifically, in cases \Rmnum{1} and \Rmnum{2}, we address the situation that there exists some $s\in\{1,2,\cdots,m-2\}$ satisfying $\alpha_{s}=0$. In case \Rmnum{3}, we treat the situation where $\alpha_{l}\neq 0$ for all $l\in\{1,2,\cdots,m-2\}$ and resolve it by using case \Rmnum{1} and case \Rmnum{2}. \par
    
    \noindent
    {\bf Case \Rmnum{1}}: Consider the case that there exists some $s\in\{1,2,\cdots,m-2\}$ such that $\alpha_{s}=0$ and $p_{s}\in [0,1)$. According to \eqref{key condition of pp:nqubit} and \eqref{eq:Tr rhoi}, for any two sets $I,J$ that satisfy $I\cup J=\{1,2,\cdots,m-2\}\backslash\{s\}$ and $I\cap J=\varnothing$, we have 
    \begin{equation}
        0={\rm Tr}(\rho_{i}A\rho_{j}\rho_{k_{s}}B)=(1-p_{s}){\rm Tr}(\rho_{i}A\rho_{j}B),
    \end{equation}
    where $A=\prod\limits_{l\in I}\rho_{k_{l}}$, $B=\prod\limits_{l\in J}\rho_{k_{l}}$ (Note: The order of operators in these two products is arbitrary). Combining this with $p_{s}\in [0,1)$, we have ${\rm Tr}(\rho_{i}A\rho_{j}B)=0$. Removing $\rho_{k_{s}}$ from $\vec{\rho}$, we obtain $m-1$ states that still contain $\rho_{i}$ and $\rho_{j}$ with $\rho_{i}\rho_{j}=0$. Since ${\rm Tr}(\rho_{i}A\rho_{j}B)=0$ for all such $A,B$, the trace of any ordered product of these $m-1$ states vanishes. Therefore, applying induction hypothesis, there exists some $t\in\{1,2,\cdots,m\}\backslash\{k_{s}\}$ such that $[\rho_{k},\rho_{l}]=0$ for all $k,l\in\{1,2,\cdots,m\}\backslash\{k_{s},t\}$. Combining this with $k_{s}\not\in\{i,j\}$, we obtain $[\rho_{i},\rho_{k_{l}}]=0$ or $[\rho_{k_{l}},\rho_{j}]=0$ for all $l$ satisfying $k_{l}\notin\{k_{s},t\}$. By calculating directly from \eqref{eq:nrho_i} and the definition of commutator, we have 
    \begin{equation}
        [\rho_{i},\rho_{k_{l}}]=[\rho_{k_{l}},\rho_{j}]=\begin{pmatrix} 0 & \alpha_{l}\\ -\overline{\alpha_{l}} & 0 \end{pmatrix}.
    \end{equation}
    Thus, $\alpha_{l}=0$ for all $l$ satisfying $k_{l}\notin\{k_{s},t\}$. Combining this with $\alpha_{s}=0$, we obtain that at most one of $\alpha_{1},\alpha_{2},\cdots,\alpha_{m-2}$ is non-zero. Hence there exists some $s\in\{1,2,\cdots,m\}$ such that $[\rho_{k},\rho_{l}]=0$ for all $k,l\in\{1,2,\cdots,m\}\backslash\{s\}$. \par

    \noindent
    {\bf Case \Rmnum{2}}: Consider the case that there exists some $s\in\{1,2,\cdots,m-2\}$ such that $\alpha_{s}=0$ and $p_{s}=1$. Based on \eqref{eq:nrho_i}, we have $\rho_{k_{s}}=\rho_{i}$ (Note: Please remember $k_{s}\notin\{i,j\}$). So $\rho_{i}\rho_{k_{s}}=\rho_{i}^{2}=\rho_{i}$. According to \eqref{key condition of pp:nqubit}, for any two sets $I,J$ that satisfy $I\cup J=\{1,2,\cdots,m-2\}\backslash\{s\}$ and $I\cap J=\varnothing$, we have 
    \begin{equation}
        0={\rm Tr}(\rho_{i}\rho_{k_{s}}A\rho_{j}B)={\rm Tr}(\rho_{i}A\rho_{j}B)
    \end{equation}
    where $A=\prod\limits_{l\in I}\rho_{k_{l}}$, $B=\prod\limits_{l\in J}\rho_{k_{l}}$. Removing $\rho_{k_{s}}$ from $\vec{\rho}$, we obtain $m-1$ states that still contain $\rho_{i}$ and $\rho_{j}$ with $\rho_{i}\rho_{j}=0$. Since ${\rm Tr}(\rho_{i}A\rho_{j}B)=0$ for all such $A,B$, the trace of any ordered product of these $m-1$ states vanishes. Therefore, applying induction hypothesis, there exists some $t\in\{1,2,\cdots,m\}\backslash\{k_{s}\}$ such that $[\rho_{k},\rho_{l}]=0$ for all $k,l\in\{1,2,\cdots,m\}\backslash\{k_{s},t\}$. 
    Combining this with $[\rho_{i},\rho_{k_{l}}]=[\rho_{k_{l}},\rho_{j}]=\begin{pmatrix} 0 & \alpha_{l}\\ -\overline{\alpha_{l}} & 0 \end{pmatrix}$ and $\alpha_{s}=0$, we obtain that at most one of $\alpha_{1},\alpha_{2},\cdots,\alpha_{m-2}$ is non-zero. Hence there exists some $s\in\{1,2,\cdots,m\}$ such that $[\rho_{k},\rho_{l}]=0$ for all $k,l\in\{1,2,\cdots,m\}\backslash\{s\}$. \par
    
    \noindent
    {\bf Case 
    \Rmnum{3}}: Consider the case where $\alpha_{l}\neq 0$ for all $l\in\{1,2,\cdots,m-2\}$. If $\alpha_{t}p_{s}+(1-p_{t})\alpha_{s}=\braket{0|\rho_{k_{s}}\rho_{k_{t}}|1}=0$ for all $s,t\in\{1,2,\cdots,m-2\}$, then $\frac{\alpha_{s}}{\alpha_{t}}=-\frac{p_{s}}{1-p_{t}}<0$ for all $s,t\in\{1,2,\cdots,m-2\}$. Thus, $1=\frac{\alpha_{1}}{\alpha_{2}}\frac{\alpha_{2}}{\alpha_{3}}\frac{\alpha_{3}}{\alpha_{1}}<0$, a contradiction. Consequently, there exist some $s,t\in\{1,2,\cdots,m-2\}$ such that $\braket{0|\rho_{k_{s}}\rho_{k_{t}}|1}\neq 0$. From \eqref{key condition of pp:nqubit}, we have 
    \begin{equation}
       0 = {\rm Tr}(\rho_{i}\rho_{k_{s}}\rho_{k_{t}}\rho_{j}B)=\braket{0|\rho_{k_{s}}\rho_{k_{t}}|1}\braket{1|B|0},
    \end{equation}
    where $B=\prod\limits_{l\in I}\rho_{k_{l}}$ and 
    \begin{equation}\label{eq:I=1,2,...,m-2 sm s,t}
        I=\{1,2,\cdots,m-2\}\backslash\{s,t\}.
    \end{equation}
    Hence $\braket{1|B|0}=0$. Combining this with
    \begin{equation}
        0 = {\rm Tr}(\rho_{i}\rho_{k_{s}}\rho_{j}\rho_{k_{t}}B)=\alpha_{s}(\overline{\alpha_{t}}\braket{0|B|0}+(1-p_{t})\braket{1|B|0}),
    \end{equation}
    we obtain $\braket{0|B|0}=0$. So $B\ket{0}=0$.\par
    We claim that there exist some $a,b\in I$ such that $\rho_{k_{a}}\rho_{k_{b}}=0$. Denote $R=\{r\in I|\det(\rho_{k_{r}})>0\}$. Then $I\backslash R=\{r\in I|\det(\rho_{k_{r}})=0\}$. If $|R|=|I|$, then  $\prod\limits_{r\in R}\rho_{k_{r}}\ket{0}=B\ket{0}=0$. Combining this with $\det(\prod\limits_{r\in R}\rho_{k_{r}})>0$, we have $\ket{0}=0$, a contradiction. Hence $|R|\leq|I|-1$. So we can see that $\prod\limits_{r\in R}\rho_{k_{r}}\prod\limits_{l\in I\backslash R}\rho_{k_{l}}\ket{0}=B\ket{0}=0$. Combining this with $\det(\prod\limits_{r\in R}\rho_{k_{r}})>0$, we have $\prod\limits_{l\in I\backslash R}\rho_{k_{l}}\ket{0}=0$. If $|R|=|I|-1$, then $\rho_{k_{l}}\ket{0}=0$ for $l\in I\backslash R$, contradicting the fact that $\alpha_{l}\neq 0$ for all $l\in I\backslash R$. Hence $|R|\leq|I|-2$. From $\prod\limits_{l\in I\backslash R}\rho_{k_{l}}\ket{0}=0$, $\alpha_{l}\neq 0$ and $\rho_{k_{l}}$ is pure state for all $l\in I\backslash R$, we obtain that there exist some $a,b\in I\backslash R\subseteq I$ such that $\rho_{k_{a}}\rho_{k_{b}}=0$. Thus, the claim is proved. \par
    Therefore, we can choose a basis $\ket{u},\ket{v}$ in which 
    \begin{equation}
    \rho_{k_{a}}=\begin{pmatrix} 1 & 0 \\ 0 & 0 \end{pmatrix},
    \rho_{k_{b}}=\begin{pmatrix} 0 & 0 \\ 0 & 1 \end{pmatrix},
    \rho_{k}=\begin{pmatrix} q_{k} & \beta_{k}\\ \overline{\beta_{k}} & 1-q_{k} \end{pmatrix}, \forall k\in \{1,2,\cdots,m\}\backslash\{k_{a},k_{b}\}
    \end{equation}
    where $q_{k}\in [0,1]$, $\beta_{k}\in \mathbb{C}$, $|\beta_{k}|^{2}\leq q_{k}(1-q_{k})$. If we know that at least one of $\beta_{k}(k\in \{1,2,\cdots,m\}\backslash\{k_{a},k_{b}\})$ is zero, then according to Case \Rmnum{1} and Case \Rmnum{2}, we can obtain there exists some $s\in\{1,2,\cdots,m\}$ such that $[\rho_{k},\rho_{l}]=0$ for all $k,l\in\{1,2,\cdots,m\}\backslash\{s\}$. Consequently, we only need consider the case where $\beta_{k}\neq 0$ for all $k\in \{1,2,\cdots,m\}\backslash\{k_{a},k_{b}\}$. \par 
    Denote $J=I\backslash\{a,b\}$. Due to $m\geq 7$, we have $J\neq\varnothing$.\par
    If $|J|=1$, we denote $\{c\}=J=I\backslash\{a,b\}$. Then 
    \begin{equation}
    \begin{split}
    \beta_{k_{c}}\braket{v|0}\ket{u}
    =\rho_{k_{a}}\rho_{k_{c}}\rho_{k_{b}}\ket{0}=B\ket{0}=0,\\
    \overline{\beta_{k_{c}}}\braket{u|0}\ket{v}
    =\rho_{k_{b}}\rho_{k_{c}}\rho_{k_{a}}\ket{0}=B\ket{0}=0.
    \end{split}
    \end{equation}
    Combining this with $\beta_{k_{c}}\neq 0$, we have $\braket{v|0}=\braket{u|0}=0$. Hence $\ket{0}=\braket{u|0}\ket{u}+\braket{v|0}\ket{v}=0$, a contradiction.\par
    If $|J|\geq 2$, then we claim that $\left(\prod\limits_{l\in J\backslash\{c\}}\rho_{k_{l}}\right)\ket{0}=0$ for any $c\in J$. Indeed, for any $c\in J$, we have 
    \begin{equation}
    \begin{split}
    \beta_{k_{c}}\braket{v|B_{1}|0}\ket{u}
    =\rho_{k_{a}}\rho_{k_{c}}\rho_{k_{b}}B_{1}\ket{0}=B\ket{0}=0,\\
    \overline{\beta_{k_{c}}}\braket{u|B_{1}|0}\ket{v}
    =\rho_{k_{b}}\rho_{k_{c}}\rho_{k_{a}}B_{1}\ket{0}=B\ket{0}=0,
    \end{split}
    \end{equation}
    where $B_{1}=\prod\limits_{l\in J\backslash\{c\}}\rho_{k_{l}}$. Combining this with $\beta_{k_{c}}\neq 0$, we obtain $\braket{u|B_{1}|0}=\braket{v|B_{1}|0}=0$. Hence $B_{1}\ket{0}=\braket{u|B_{1}|0}\ket{u}+\braket{v|B_{1}|0}\ket{v}=0$. Thus, the claim is proved. \par
    When $|J|=2$. Due to $\left(\prod\limits_{l\in J\backslash\{c\}}\rho_{k_{l}}\right)\ket{0}=0$ for any $c\in J$, we have $\begin{pmatrix} p_{r}\\ \overline{\alpha}_{r}\end{pmatrix}=\rho_{k_{r}}\ket{0}=B_{1}\ket{0}=0$ for any $r\in J$, contradicting the fact that $\alpha_{l}\neq 0$ for all $l\in\{1,2,\cdots,m-2\}$.\par
    When $|J|=3$. Due to $\left(\prod\limits_{l\in J\backslash\{c\}}\rho_{k_{l}}\right)\ket{0}=0$ for any $c\in J$, we have $\rho_{k_{c}}\rho_{k_{d}}\ket{0}=0$ for any two different number $c,d\in J$. Combining this with $\alpha_{l}\neq 0$ for all $l\in J$, we have $\rho_{k_{c}}\rho_{k_{d}}=0$ for any two different number $c,d\in J$. Hence $\frac{\alpha_{c}}{\alpha_{d}}=-\frac{p_{c}}{1-p_{d}}<0$ for any two different number $c,d\in J$. Denote $\{d_{1},d_{2},d_{3}\}=J$, then $1=\frac{\alpha_{d_{1}}}{\alpha_{d_{2}}}\frac{\alpha_{d_{2}}}{\alpha_{d_{3}}}\frac{\alpha_{d_{3}}}{\alpha_{d_{1}}}<0$, a contradiction.\par
    When $|J|\geq 4$. Take $c\in J$ and denote $I_{1}=J\backslash\{c\}$. Then $B_{1}\ket{0}=0$, where $B_{1}=\prod\limits_{l\in I_{1}}\rho_{k_{l}}$. Hence there exist some $a_{1},b_{1}\in I_{1}$ such that $\rho_{k_{a_{1}}}\rho_{k_{b_{1}}}=0$. So we can choose a basis $\ket{u_{1}},\ket{v_{1}}$ in which 
    \begin{equation}
    \rho_{k_{a_{1}}}=\begin{pmatrix} 1 & 0 \\ 0 & 0 \end{pmatrix},
    \rho_{k_{b_{1}}}=\begin{pmatrix} 0 & 0 \\ 0 & 1 \end{pmatrix},
    \rho_{k}=\begin{pmatrix} q_{k}^{(1)} & \beta_{k}^{(1)}\\ \overline{\beta_{k}^{(1)}} & 1-q_{k}^{(1)} \end{pmatrix}, \forall k\in \{1,2,\cdots,m\}\backslash\{k_{a_{1}},k_{b_{1}}\}
    \end{equation}
    where $q_{k}^{(1)}\in [0,1]$, $\beta_{k}^{(1)}\in \mathbb{C}$, $|\beta_{k}^{(1)}|^{2}\leq q_{k}^{(1)}(1-q_{k}^{(1)})$. If we know that at least one of $\beta_{k}^{(1)}(k\in \{1,2,\cdots,m\}\backslash\{k_{a_1},k_{b_1}\})$ is zero, then according to Case \Rmnum{1} and Case \Rmnum{2}, we can obtain there exists some $s\in\{1,2,\cdots,m\}$ such that $[\rho_{k},\rho_{l}]=0$ for all $k,l\in\{1,2,\cdots,m\}\backslash\{s\}$. Consequently, we only need consider the case where $\beta_{k}^{(1)}\neq 0$ for all $k\in \{1,2,\cdots,m\}\backslash\{k_{a_{1}},k_{b_{1}}\}$. Denote $J_{1}=I_{1}\backslash\{a_{1},b_{1}\}=J\backslash\{c,a_{1},b_{1}\}$. Then $|J_{1}|\geq 1$. Therefore, we can repeat the previous discussion of $J$ for $J_ {1}$. By repeating above process, we can complete the proof of (\rmnum{1})$\Rightarrow$(\rmnum{2}).\par
    
    (\rmnum{2})$\Rightarrow$(\rmnum{1}). Assume that $\{1,2,\cdots,n\}\backslash\{s\}=\{k_{1},k_{2},\cdots,k_{n-1}\}$, where $k_{1}<k_{2}<\cdots<k_{n-1}$. Due to $[\rho_{k},\rho_{l}]=0$ for all $k,l\in\{k_{1},k_{2},\cdots,k_{n-1}\}$, we can choose a basis $\ket{0},\ket{1}$ in which
    \begin{equation}
        \rho_{s}=\begin{pmatrix} p & \beta\\ \overline{\beta} & 1-p \end{pmatrix},
        \rho_{k_{l}}=\begin{pmatrix} a_{l} & 0 \\ 0 & 1-a_{l} \end{pmatrix}, l=1,2,\cdots,n-1
    \end{equation}
    where $a_{l},p\in [0,1]$, $\beta\in \mathbb{C}$ and $|\beta|^{2}\leq p(1-p)$. Due to cyclicity of the trace, we see that \eqref{key condition of pp:nqubit} is equivalent to ${\rm Tr}(\rho_{\sigma(1)}\rho_{\sigma(2)}\cdots\rho_{\sigma(n)})=0$ for all $\sigma\in S_{n}$. So we obtain $0={\rm Tr}(\rho_{k_{1}}\rho_{k_{2}}\cdots\rho_{k_{n-1}}\rho_{s})=p\prod\limits_{l=1}^{n-1}a_{l}+(1-p)\prod\limits_{l=1}^{n-1}(1-a_{l})$. Therefore, we obtain $p\prod\limits_{l=1}^{n-1}a_{l}=0$ and $(1-p)\prod\limits_{l=1}^{n-1}(1-a_{l})=0$. If $p\in (0,1)$, then $\prod\limits_{l=1}^{n-1}a_{l}=0$ and $\prod\limits_{l=1}^{n-1}(1-a_{l})=0$. So there exist $i,j\in\{1,2,\cdots,n-1\}(i\neq j)$ such that $a_{i}=0$ and $a_{j}=1$. This implies $\rho_{k_{i}}\rho_{k_{j}}=0$. If $p=0$ or $p=1$, then $\beta=0$ from $|\beta|^{2}\leq p(1-p)$. Clearly, there exists some $i\in\{1,2,\cdots,n-1\}$ such that $a_{i}=1$ or $a_{i}=0$. Hence $\rho_{k_{i}}\rho_{s}=0$.\qed

\section{The proof of Theorem \ref{nimaginarity}}\label{Appendix nimaginarity}
To prove Theorem \ref{nimaginarity}, we need to first prove the following lemma.
\begin{lemma}\label{567imaginarity}
    Let $\vec{\rho}=(\rho_{1},\cdots,\rho_{n})$ be an $n$-tuple of single qubit states where $\rho_{i}\in \mathcal{D}(\mathbb{C}^{2})$ and $n\in\{5,6,7\}$. Denote the Bargmann scenario $\mathcal{W}=\{(\tau(1),\tau(2),\cdots,\tau(n))|\tau\in S_{n}\enspace\mbox{and}\enspace\tau(1)=1\}$, where $S_{n}$ is the symmetric group on $\{1,2,\cdots,n\}$. Suppose that 
    \begin{equation}\label{keyeq-567imaginarity}
       \Delta_{w}(\vec{\rho})=0,\enspace\forall w\in\mathcal{W}.
    \end{equation}
    and ${\rm Tr}(\rho_{k}^{2})>\frac{1}{2}$ for all $k\in\{1,2,\cdots,n\}$. If the tuple $\vec{\rho}=(\rho_{1},\rho_{2},\cdots,\rho_{n})$ is imaginarity-free, then there exists some $i,j\in\{1,2,\cdots,n\}$ such that ${\rm Tr}(\rho_{i}\rho_{j})=0$.
\end{lemma}
\begin{proof}
    Write these states in Bloch form, 
    \begin{equation}
    \rho_{k}=\frac{\mathbb{I}_{2}+\boldsymbol{r}_{k}\cdot\boldsymbol{\sigma}}{2}
    \end{equation}
    with $|\boldsymbol{r}_{k}|\leq 1$ ($k=1,2,\cdots,n$). For all $k\in\{1,2,\cdots,n\}$, we have $\boldsymbol{r}_{k}\neq 0$ because ${\rm Tr}(\rho_{k}^{2})\neq\frac{1}{2}$.\par
    Due to $\vec{\rho}=(\rho_{1},\rho_{2},\cdots,\rho_{n})$ is imaginarity-free, we can assume without loss of generality that $\rho_{k}$ is real for all $k\in\{1,2,\cdots,n\}$. In this case, ${\rm Tr}(\rho_{k}\sigma_{y})=0$ for all $k\in\{1,2,\cdots,n\}$. Consequently, there exist $r_{k,1},r_{k,2}\in\mathbb{R}$ such that $\boldsymbol{r}_{k}=(r_{k,1},0,r_{k,2})$. Denote $\boldsymbol{u}_{k}=(r_{k,1},r_{k,2})^{\mathsf{T}}$ for all $k\in\{1,2,\cdots,n\}$. Clearly, $\boldsymbol{u}_{k}\neq 0$ because $\boldsymbol{r}_{k}\neq 0$. For any $k,l\in\{1,2,\cdots,n\}$, we have
    \begin{equation}
        \boldsymbol{r}_{k}\times\boldsymbol{r}_{l}=(0,r_{l,1}r_{k,2}-r_{k,1}r_{l,2},0)^{\mathsf{T}}=(0,-\det(\boldsymbol{u}_{k},\boldsymbol{u}_{l}),0)^{\mathsf{T}}.
    \end{equation}
    Combining this with $(\boldsymbol{a}\cdot\boldsymbol{\sigma})(\boldsymbol{b}\cdot\boldsymbol{\sigma})=(\boldsymbol{a}\cdot\boldsymbol{b})\mathbb{I}_{2}+i(\boldsymbol{a}\times\boldsymbol{b})\cdot\boldsymbol{\sigma}$, we obtain
    \begin{equation}\label{1commutator}
        [\rho_{k},\rho_{l}]=\frac{1}{4}[\boldsymbol{r}_{k}\cdot\boldsymbol{\sigma},\boldsymbol{r}_{l}\cdot\boldsymbol{\sigma}]=\frac{\mathrm{i}}{4}(\boldsymbol{r}_{k}\times\boldsymbol{r}_{l}-\boldsymbol{r}_{l}\times\boldsymbol{r}_{k})\cdot\boldsymbol{\sigma}=\frac{\mathrm{i}}{2}(\boldsymbol{r}_{k}\times\boldsymbol{r}_{l})\cdot\boldsymbol{\sigma}=-\frac{\mathrm{i}}{2}\det(\boldsymbol{u}_{k},\boldsymbol{u}_{l})\sigma_{y}.
    \end{equation}
    Therefore, for any $k_{1},l_{1},k_{2},l_{2}\in\{1,2,\cdots,n\}$, we have
    \begin{equation}\label{2commutator567}
    [\rho_{k_{1}},\rho_{l_{1}}][\rho_{k_{2}},\rho_{l_{2}}] =-\frac{1}{4}\det(\boldsymbol{u}_{k_{1}},\boldsymbol{u}_{l_{1}})\det(\boldsymbol{u}_{k_{2}},\boldsymbol{u}_{l_{2}})\mathbb{I}_{2}.
    \end{equation}
    
    First, let us prove Lemma \ref{567imaginarity} for the case $n=5$. Due to cyclicity of the trace, we see that \eqref{keyeq-567imaginarity} is equivalent to ${\rm Tr}(\rho_{\tau(1)}\rho_{\tau(2)}\rho_{\tau(3)}\rho_{\tau(4)}\rho_{\tau(5)})=0$ for all $\tau\in S_{5}$. Combining this with \eqref{2commutator567}, we obtain 
    \begin{equation}
    \det(\boldsymbol{u}_{\tau(1)},\boldsymbol{u}_{\tau(2)})
    \det(\boldsymbol{u}_{\tau(3)},\boldsymbol{u}_{\tau(4)})
    =-4{\rm Tr}([\rho_{\tau(1)},\rho_{\tau(2)}][\rho_{\tau(3)},\rho_{\tau(4)}]\rho_{\tau(5)})=0
    \end{equation}
    for all $\tau\in S_{5}$. If there exist some $j,k,l\in\{1,2,3,4,5\}$ such that $\boldsymbol{u}_{j},\boldsymbol{u}_{k},\boldsymbol{u}_{l}$ are pairwise non-collinear, then we have 
    \begin{equation}
        \det(\boldsymbol{u}_{j},\boldsymbol{u}_{k})
    \det(\boldsymbol{u}_{l},\boldsymbol{u}_{m})\neq 0 \enspace\mbox{or}\enspace \det(\boldsymbol{u}_{j},\boldsymbol{u}_{m})
    \det(\boldsymbol{u}_{k},\boldsymbol{u}_{l})\neq 0
    \end{equation}
    for any $m\in\{1,2,3,4,5\}\backslash\{j,k,l\}$, a contradiction. From this and $\boldsymbol{u}_{k}\neq 0$ for all $k\in\{1,2,3,4,5\}$, we obtain that at least four vectors in the sequence $\{\boldsymbol{u}_{k}\}_{k=1}^{5}$ are collinear, which is equivalent to at least four vectors in the sequence $\{\boldsymbol{r}_{k}\}_{k=1}^{5}$ are collinear. Therefore, there exists some $s\in\{1,2,3,4,5\}$ such that $[\rho_{k},\rho_{l}]=0$ for all $k,l\in\{1,2,3,4,5\}\backslash\{s\}$. From Proposition \ref{thm_5rho}, there exists some $i,j\in\{1,2,3,4,5\}$ such that ${\rm Tr}(\rho_{i}\rho_{j})=0$.\par

    Next, let us prove Lemma \ref{567imaginarity} for the case $n=6,7$. We define $G=\{\{s,t\}|\det(\boldsymbol{u}_{s},\boldsymbol{u}_{t})\neq 0\}$. Clearly, we have $G=\varnothing$ or $|G|\geq 2$. If $G=\varnothing$, then according to Theorem \ref{thm_nrho}, there exists some $i,j\in\{1,2,\cdots,n\}$ such that ${\rm Tr}(\rho_{i}\rho_{j})=0$. Thus, we assume that $|G|\geq 2$ below.\par
    \noindent
    {\bf Case \Rmnum{1}}: Consider the case that there exist two different members $E_{1},E_{2}\in G$ satisfy $E_{1}\cap E_{2}=\varnothing$. \par
    \noindent
    (\rmnum{1}) If $n=6$, then there exist some $\tau\in S_{6}$ such that $E_{1}=\{\tau(1),\tau(2)\}$ and $E_{2}=\{\tau(3),\tau(4)\}$. Due to cyclicity of the trace, we see that \eqref{keyeq-567imaginarity} is equivalent to ${\rm Tr}(\rho_{\tau(1)}\rho_{\tau(2)}\rho_{\tau(3)}\rho_{\tau(4)}\rho_{\tau(5)}\rho_{\tau(6)})=0$ for all $\tau\in S_{6}$. With help of this and \eqref{2commutator567}, we obtain 
    \begin{equation*}
    \det(\boldsymbol{u}_{\tau(1)},\boldsymbol{u}_{\tau(2)})
    \det(\boldsymbol{u}_{\tau(3)},\boldsymbol{u}_{\tau(4)})
    {\rm Tr}(\rho_{\tau(5)}\rho_{\tau(6)})=-4{\rm Tr}([\rho_{\tau(1)},\rho_{\tau(2)}][\rho_{\tau(3)},\rho_{\tau(4)}]\rho_{\tau(5)}\rho_{\tau(6)})=0.
    \end{equation*} 
    Combining this with $\{\tau(1),\tau(2)\}\in G$ and $\{\tau(3),\tau(4)\}\in G$, we have ${\rm Tr}(\rho_{\tau(5)}\rho_{\tau(6)})=0$. \par
    \noindent
    (\rmnum{2}) If $n=7$, then there exist some $\tau\in S_{7}$ such that $E_{1}=\{\tau(1),\tau(2)\}$ and $E_{2}=\{\tau(3),\tau(4)\}$. According to \eqref{1commutator}, for any $k_{1},l_{1},k_{2},l_{2},m\in\{1,2,\cdots,7\}$, we have
    \begin{equation}\label{2commutator7}
    \begin{split}
    [\rho_{k_{1}},\rho_{l_{1}}][\rho_{k_{2}},\rho_{l_{2}}] & =-\frac{1}{4}\det(\boldsymbol{u}_{k_{1}},\boldsymbol{u}_{l_{1}})\det(\boldsymbol{u}_{k_{2}},\boldsymbol{u}_{l_{2}})\mathbb{I}_{2},\\
    [\rho_{k_{1}},\rho_{l_{1}}]\rho_{m}[\rho_{k_{2}},\rho_{l_{2}}] & =-\frac{1}{4}\det(\boldsymbol{u}_{k_{1}},\boldsymbol{u}_{l_{1}})\det(\boldsymbol{u}_{k_{2}},\boldsymbol{u}_{l_{2}})\sigma_{y}\rho_{m}\sigma_{y}.    
    \end{split}
    \end{equation}
    Due to cyclicity of the trace, we see that \eqref{keyeq-567imaginarity} is equivalent to ${\rm Tr}(\rho_{\tau(1)}\rho_{\tau(2)}\cdots\rho_{\tau(7)})=0$ for all $\tau\in S_{7}$. Thus, we have 
    \begin{equation}
    \begin{split}
    -4{\rm Tr}([\rho_{\tau(1)},\rho_{\tau(2)}][\rho_{\tau(3)},\rho_{\tau(4)}]\rho_{\tau(5)}\rho_{\tau(6)}\rho_{\tau(7)})=0,\\
    -4{\rm Tr}([\rho_{\tau(1)},\rho_{\tau(2)}]\rho_{\tau(5)}[\rho_{\tau(3)},\rho_{\tau(4)}]\rho_{\tau(6)}\rho_{\tau(7)})=0.
    \end{split}
    \end{equation}
    With help of this and \eqref{2commutator7}, we obtain 
    \begin{equation}\label{keydet7}
    \begin{split}
    \det(\boldsymbol{u}_{\tau(1)},\boldsymbol{u}_{\tau(2)})
    \det(\boldsymbol{u}_{\tau(3)},\boldsymbol{u}_{\tau(4)})
    {\rm Tr}(\rho_{\tau(5)}\rho_{\tau(6)}\rho_{\tau(7)}) & =0,\\
    \det(\boldsymbol{u}_{\tau(1)},\boldsymbol{u}_{\tau(2)})
    \det(\boldsymbol{u}_{\tau(3)},\boldsymbol{u}_{\tau(4)})
    {\rm Tr}(\sigma_{y}\rho_{\tau(5)}\sigma_{y}\rho_{\tau(6)}\rho_{\tau(7)}) & =0.  
    \end{split}
    \end{equation}
    Combining this with $\{\tau(1),\tau(2)\}\in G$ and $\{\tau(3),\tau(4)\}\in G$, we have
    \begin{equation}\label{keyTr7}
    \begin{cases}
    {\rm Tr}(\rho_{\tau(5)}\rho_{\tau(6)}\rho_{\tau(7)})=0, \\
    {\rm Tr}(\sigma_{y}\rho_{\tau(5)}\sigma_{y}\rho_{\tau(6)}\rho_{\tau(7)})=0.
    \end{cases}
    \end{equation}
    Due to $\boldsymbol{r}_{k}=(r_{k,1},0,r_{k,2})$, we have $\sigma_{y}\rho_{k}\sigma_{y}=\frac{1}{2}(\mathbb{I}_{2}-\boldsymbol{r}_{k}\cdot\boldsymbol{\sigma})$. Combining this with \eqref{keyTr7}, we obtain
    \begin{equation*}
    \begin{cases}
    0= 1+\boldsymbol{r}_{\tau(5)}\cdot\boldsymbol{r}_{\tau(6)}+\boldsymbol{r}_{\tau(5)}\cdot\boldsymbol{r}_{\tau(7)}+\boldsymbol{r}_{\tau(6)}\cdot\boldsymbol{r}_{\tau(7)}\\
    0= 1-\boldsymbol{r}_{\tau(5)}\cdot\boldsymbol{r}_{\tau(6)}-\boldsymbol{r}_{\tau(5)}\cdot\boldsymbol{r}_{\tau(7)}+\boldsymbol{r}_{\tau(6)}\cdot\boldsymbol{r}_{\tau(7)}.
    \end{cases}
    \end{equation*}
    Therefore, we have $\boldsymbol{r}_{\tau(6)}\cdot\boldsymbol{r}_{\tau(7)}=-1$. Hence ${\rm Tr}(\rho_{\tau(6)}\rho_{\tau(7)})=\frac{1}{2}(1+\boldsymbol{r}_{\tau(6)}\cdot\boldsymbol{r}_{\tau(7)})=0$. \par
    
    \noindent
    {\bf Case \Rmnum{2}}: Consider the case that $E_{1}\cap E_{2}\neq\varnothing$ for all  $E_{1},E_{2}\in G$. Combining this with $|G|\geq 2$, we can take $\{b,a\},\{a,c\}\in G$. For any $d\in\{1,2,\cdots,n\}\backslash\{a,b,c\}$, according to $E_{1}\cap E_{2}\neq\varnothing$ for all $E_{1},E_{2}\in G$, we have $\{b,d\}\not\in G$ and $\{c,d\}\not\in G$ because $\{b,d\}\cap\{a,c\}=\varnothing$ and $\{c,d\}\cap\{b,a\}=\varnothing$. Combining this with $\boldsymbol{u}_{k}\neq 0$ for all $k\in\{1,2,\cdots,n\}$, we obtain $\boldsymbol{u}_{d},\boldsymbol{u}_{b},\boldsymbol{u}_{c}$ are collinear. Consequently, at least $n-1$ vectors in the sequence $\{\boldsymbol{u}_{k}\}_{k=1}^{n}$ are collinear, which is equivalent to at least $n-1$ vectors in the sequence $\{\boldsymbol{r}_{k}\}_{k=1}^{n}$ are collinear. Therefore, there exists some $s\in\{1,2,\cdots,n\}$ such that $[\rho_{k},\rho_{l}]=0$ for all $k,l\in\{1,2,\cdots,n\}\backslash\{s\}$. From Theorem \ref{thm_nrho}, there exists some $i,j\in\{1,2,\cdots,n\}$ such that ${\rm Tr}(\rho_{i}\rho_{j})=0$.
\end{proof}

    Now, we can prove Theorem \ref{nimaginarity}.

{\bf The proof of Theorem \ref{nimaginarity}}: Write these states in Bloch form, 
    \begin{equation}
    \rho_{k}=\frac{\mathbb{I}_{2}+\boldsymbol{r}_{k}\cdot\boldsymbol{\sigma}}{2}
    \end{equation}
    with $|\boldsymbol{r}_{k}|\leq 1$ ($k=1,2,\cdots,n$). For all $k\in\{1,2,\cdots,n\}$, we have $\boldsymbol{r}_{k}\neq 0$ because ${\rm Tr}(\rho_{k}^{2})\neq\frac{1}{2}$.\par
    \noindent
    (\rmnum{1})$\Rightarrow$(\rmnum{2}). Note that ${\rm Tr}(\rho_{i}\rho_{j})=0$ is equivalent to $\rho_{i}\rho_{j}=0$. From Propositions \ref{pp:rho1rho2rho3rho4}, \ref{thm_5rho} and \ref{thm_nrho}, there exists some $s\in\{1,2,\cdots,n\}$ such that $[\rho_{k},\rho_{l}]=0$ for all $k,l\in\{1,2,\cdots,n\}\backslash\{s\}$. So at least $n-1$ vectors in the sequence $\{\boldsymbol{r}_{k}\}_{k=1}^{n}$ are collinear. Denote $B=(\boldsymbol{r}_{1},\boldsymbol{r}_{2},\cdots,\boldsymbol{r}_{n})$. Then we have ${\rm rank}(G_{\boldsymbol{r}_{\vec{\rho}}})={\rm rank}(B^{\mathsf{T}}B)={\rm rank}(B)\leq 2$ where $G_{\boldsymbol{r}_{\vec{\rho}}}=(\braket{\boldsymbol{r}_{k},\boldsymbol{r}_{l}})_{n\times n}=B^TB$. Therefore, the tuple $\vec{\rho}=(\rho_{1},\rho_{2},\cdots,\rho_{n})$ is imaginarity-free from Lemma \ref{th:imaginarity}.\par

    \noindent
    (\rmnum{2})$\Rightarrow$(\rmnum{1}). From Proposition \ref{4imaginarity} and Lemma \ref{567imaginarity}, we have proven the claim (\rmnum{2})$\Rightarrow$(\rmnum{1}) for $n=4,5,6,7$. We assume $n\geq 8$ below. To apply the induction on $n$, we assume that the claim is true for any $m$-tuple of single qubit states when $m\in\{4,5,\cdots,n-1\}$. We need prove the claim for $n$.\par
    
    Due to $\vec{\rho}=(\rho_{1},\rho_{2},\cdots,\rho_{n})$ is imaginarity-free, we can assume without loss of generality that $\rho_{k}$ is real for all $k\in\{1,2,\cdots,n\}$. In this case, ${\rm Tr}(\rho_{k}\sigma_{y})=0$ for all $k\in\{1,2,\cdots,n\}$. Consequently, there exist $r_{k,1},r_{k,2}\in\mathbb{R}$ such that $\boldsymbol{r}_{k}=(r_{k,1},0,r_{k,2})$. Denote $\boldsymbol{u}_{k}=(r_{k,1},r_{k,2})^{\mathsf{T}}$ for all $k\in\{1,2,\cdots,n\}$. Clearly, $\boldsymbol{u}_{k}\neq 0$ because $\boldsymbol{r}_{k}\neq 0$. For any $k,l\in\{1,2,\cdots,n\}$, we have
    \begin{equation}
        \boldsymbol{r}_{k}\times\boldsymbol{r}_{l}=(0,r_{l,1}r_{k,2}-r_{k,1}r_{l,2},0)^{\mathsf{T}}=(0,-\det(\boldsymbol{u}_{k},\boldsymbol{u}_{l}),0)^{\mathsf{T}}.
    \end{equation}
    Combining this with $(\boldsymbol{a}\cdot\boldsymbol{\sigma})(\boldsymbol{b}\cdot\boldsymbol{\sigma})=(\boldsymbol{a}\cdot\boldsymbol{b})\mathbb{I}_{2}+i(\boldsymbol{a}\times\boldsymbol{b})\cdot\boldsymbol{\sigma}$, we obtain
    \begin{equation}
        [\rho_{k},\rho_{l}]=\frac{1}{4}[\boldsymbol{r}_{k}\cdot\boldsymbol{\sigma},\boldsymbol{r}_{l}\cdot\boldsymbol{\sigma}]=\frac{\mathrm{i}}{4}(\boldsymbol{r}_{k}\times\boldsymbol{r}_{l}-\boldsymbol{r}_{l}\times\boldsymbol{r}_{k})\cdot\boldsymbol{\sigma}=\frac{\mathrm{i}}{2}(\boldsymbol{r}_{k}\times\boldsymbol{r}_{l})\cdot\boldsymbol{\sigma}=-\frac{\mathrm{i}}{2}\det(\boldsymbol{u}_{k},\boldsymbol{u}_{l})\sigma_{y}.
    \end{equation}
    Therefore, for any $k_{1},l_{1},k_{2},l_{2}\in\{1,2,\cdots,n\}$, we have
    \begin{equation}\label{2commutatorn}
    [\rho_{k_{1}},\rho_{l_{1}}][\rho_{k_{2}},\rho_{l_{2}}] =-\frac{1}{4}\det(\boldsymbol{u}_{k_{1}},\boldsymbol{u}_{l_{1}})\det(\boldsymbol{u}_{k_{2}},\boldsymbol{u}_{l_{2}})\mathbb{I}_{2}.
    \end{equation}
    Due to cyclicity of the trace, we see that \eqref{keyeq-nimaginarity} is equivalent to ${\rm Tr}(\rho_{\tau(1)}\rho_{\tau(2)}\cdots\rho_{\tau(n)})=0$ for all $\tau\in S_{n}$. Thus, for any $\tau\in S_{n}$, we have 
    \begin{equation}
        -4{\rm Tr}([\rho_{\tau(1)},\rho_{\tau(2)}][\rho_{\tau(3)},\rho_{\tau(4)}]\rho_{\tau(5)}\cdots\rho_{\tau(n)})=0.
    \end{equation}
    Combining this with \eqref{2commutatorn}, we obtain 
    \begin{equation}\label{keydetn}
    \det(\boldsymbol{u}_{\tau(1)},\boldsymbol{u}_{\tau(2)})
    \det(\boldsymbol{u}_{\tau(3)},\boldsymbol{u}_{\tau(4)})
    {\rm Tr}(\rho_{\tau(5)}\cdots\rho_{\tau(n)})=0.
    \end{equation}
    We define $G=\{\{s,t\}|\det(\boldsymbol{u}_{s},\boldsymbol{u}_{t})\neq 0\}$. Clearly, we have $G=\varnothing$ or $|G|\geq 2$. From Theorem \ref{thm_nrho}, there exists some $i,j\in\{1,2,\cdots,n\}$ such that ${\rm Tr}(\rho_{i}\rho_{j})=0$ when $G=\varnothing$. Thus, we assume that $|G|\geq 2$ below.\par

    \noindent
    {\bf Case \Rmnum{1}}: Consider the case that there exist two different members $E_{1},E_{2}\in G$ satisfy $E_{1}\cap E_{2}=\varnothing$. Assume without loss of generality that $E_{1}=\{1,2\}$ and $E_{2}=\{3,4\}$. Combining this with \eqref{keydetn}, we have 
    \begin{equation}
    {\rm Tr}(\rho_{\tau(5)}\rho_{\tau(6)}\cdots\rho_{\tau(n)})=0
    \end{equation}
    for all $\tau\in S_{\{5,6,\cdots,n\}}$, where $S_{\{5,6,\cdots,n\}}$ is the symmetric group on $\{5,6,\cdots,n\}$. Clearly, $|\{5,6,\cdots,n\}|=n-4$. Applying induction hypothesis, there exists some $i,j\in\{5,6,\cdots,n\}\subseteq\{1,2,\cdots,n\}$ such that ${\rm Tr}(\rho_{i}\rho_{j})=0$. \par

    \noindent
    {\bf Case \Rmnum{2}}: Consider the case that $E_{1}\cap E_{2}\neq\varnothing$ for all  $E_{1},E_{2}\in G$. Combining this with $|G|\geq 2$, we can take $\{b,a\},\{a,c\}\in G$. For any $d\in\{1,2,\cdots,n\}\backslash\{a,b,c\}$, according to $E_{1}\cap E_{2}\neq\varnothing$ for all $E_{1},E_{2}\in G$, we have $\{b,d\}\not\in G$ and $\{c,d\}\not\in G$ because $\{b,d\}\cap\{a,c\}=\varnothing$ and $\{c,d\}\cap\{b,a\}=\varnothing$. Combining this with $\boldsymbol{u}_{k}\neq 0$ for all $k\in\{1,2,\cdots,n\}$, we obtain $\boldsymbol{u}_{d},\boldsymbol{u}_{b},\boldsymbol{u}_{c}$ are collinear. Consequently, at least $n-1$ vectors in the sequence $\{\boldsymbol{u}_{k}\}_{k=1}^{n}$ are collinear, which is equivalent to at least $n-1$ vectors in the sequence $\{\boldsymbol{r}_{k}\}_{k=1}^{n}$ are collinear. Therefore, there exists some $s\in\{1,2,\cdots,n\}$ such that $[\rho_{k},\rho_{l}]=0$ for all $k,l\in\{1,2,\cdots,n\}\backslash\{s\}$. From Theorem \ref{thm_nrho}, we can finish the proof. \qed

\end{appendices}

\end{document}